\documentclass[12pt, draftclsnofoot, onecolumn]{IEEEtran}
\usepackage{amsfonts}
\usepackage{dsfont}
\usepackage[normalem]{ulem}
\usepackage{setspace}
\usepackage{color}
\usepackage{amssymb}
\usepackage{cite}
\usepackage[cmex10]{amsmath} 
\usepackage{algorithm}
\usepackage{algorithmic} 
\usepackage{array}
\usepackage{mathrsfs}
\usepackage{graphicx}
\usepackage{latexsym}
\usepackage{amscd}
\usepackage{amsfonts}
\usepackage{subfigure}
\usepackage{amsmath,amscd,amssymb,verbatim}
\usepackage{graphics}
\usepackage{amsthm}
\usepackage[utf8]{inputenc}
\usepackage{authblk}
\usepackage{amsmath,amsthm}
\usepackage{url}
\usepackage{witharrows}

\newtheorem{lemma}{{Lemma}}

\begin{document}
\title{UAV-Enabled ISAC: Towards On-Demand Sensing Services and Enhanced Communication
}
\author[$\dag$]{Xiaopeng Yuan, Peng Wu, Xinran Wang, Yulin Hu, and Anke Schmeink
\thanks{X. Yuan and A. Schmeink are with Chair of Information Theory and Data Analytics, RWTH Aachen University, 52074 Aachen, Germany. (Email: $yuan|schmeink$@inda.rwth-aachen.de).} 
\thanks{P. Wu, X. Wang and Y. Hu are with School of  Electronic Information, Wuhan University, 430072 Wuhan, China.
(Email: $peng.wu|xinran.wang|yulin.hu$@whu.edu.cn).
}
\vspace{-.5cm}
}
\maketitle
 
\vspace{-1.cm}

\begin{abstract}
In this paper, we investigate an integrated sensing-and-communication (ISAC) network enabled by an unmanned aerial vehicle (UAV). The UAV is supposed to fly along a periodical circular trajectory at a fixed height for ISAC service supply from the sky. We consider on-demand sensing services, where on-demand detection and on-demand localization requests may be activated at any time toward any position within the targeted serving region. While guaranteeing satisfactory accuracy for both on-demand sensing tasks, we aim at maximizing the minimum achievable throughput among all communication users, via joint optimizing the UAV trajectory and communication user scheduling. To address the complicated problem with infinite sensing constraints, 
we characterize the on-demand detection  constraint as a restricted deployment area for UAV and the on-demand localization constraint as Cramér-Rao Bound (CRB) constraints over finite reference target points, based on which 
the original problem is simplified to more tractable one.  Afterwards, particularly aiming to ensure no violations of CRB constraints, we propose a convex approximation for the reformulated problem, where tight approximation is guaranteed at given local solution. The construction strategy for convex problem approximation allows an efficient iterative algorithm with verified convergence to a superior suboptimal solution. At last, with simulations,  we verified the applicability of our developed optimization scheme in strictly fulfilling the on-demand sensing constraints and the effectiveness of our proposed solution for simultaneously enhancing the communication throughput in  UAV-enabled ISAC.%
\end{abstract}
\vspace{-.5cm}
\vspace{-.1cm}
\section{Introduction}

With the advancement of wireless networks, 
integrated sensing and communication (ISAC) technologies have emerged as a critical driver for numerous applications, like autonomous vehicles, smart home and extended reality (XR)~\cite{isac_1}, and thus been identified as one of the key enabling technologies for the 6G era~\cite{isac_2}. Particularly, instead of independent deployments of mobile radio and radar systems in conventional networks, the integration of sensing and communication brings huge benefits~\cite{isac_3}, including low hardware cost, reduced power consumption and signaling latency. 
Towards fully leveraging these substantial advantages of ISAC, enormous research interests have arisen, e.g., in modulation and coding scheme designs~\cite{isac_4,isac_5}, and waveform designs~\cite{isac_6,isac_7,isac_8}. Extensively, the works \cite{isac_9,isac_10} investigated advanced beamforming strategies, together with smart radio environment management via reconfigurable intelligent surfaces (RIS), 
for the promotion of ISAC networks. 

In accordance with specific application demands, various sensing tasks have also been widely studied, such as target detections, distance estimation and scene recognition~\cite{isac_11}. Among them, target localization plays an important role for the future wireless networks~\cite{local_1}, enabling accurate wireless service delivery and precise target monitoring. Integrating target localization into communication, numerous researches have been conducted with respect to advanced multiple access schemes~\cite{local_2}, effective beamforming designs~\cite{local_3} and efficient resource allocation~\cite{local_4}, to deal with the fundamental tradeoff between localization accuracy and communication performance.

Furthermore, with access capability from the sky, unmanned aerial vehicles (UAVs) have shown tremendous promotion potentials in multifarious wireless applications~\cite{uav_1,add_1}. Within UAV-enabled ISAC network, the superior air-to-ground links will efficiently facilitate both communication and sensing services, and elevate the integration gain of ISAC~\cite{add_2}. Considering signal-to-noise ratio (SNR) based sensing metric, the works \cite{uav_2,uav_3,uav_4} investigated joint UAV trajectory and beamforming designs for ISAC performance enhancement, while a joint UAV trajectory and resource allocation design has been proposed in~\cite{uav_5} for a multi-target scenario. 
Referring to target localization services, 
achieving a high localization accuracy requires the collaboration of multiple stations/agents~\cite{isac_1,local_1}, employing the corresponding high spatial diversity. 
In UAV-enabled ISAC networks, such spatial diversity can be realized via the flexible deployment of UAV, making UAV a promising technique for the wireless positioning services~\cite{uav_6}. 
Nevertheless, since the localization accuracy  depends on the varying UAV position,  the existing works \cite{uav_2,uav_3,uav_4,uav_5} purely relying on an SNR-based sensing metric, become inapplicable for UAV-enabled ISAC with localization services.

To deal with this issue, two recent works \cite{Jing_1,Gu_1} have derived out the Cramér-Rao Bound (CRB) of target localization in UAV-enabled ISAC,  and proposed UAV trajectory designs respectively maximizing the energy efficiency and localization accuracy. However, to tackle the highly sophisticated localization CRB, 
both works \cite{Jing_1,Gu_1} have adopted the first-order Taylor approximation for CRB and \emph{approximately optimized} the  objectives. 
Such simplex approximation introduces considerable approximation errors in CRB evaluations. Especially when the UAV-enabled ISAC network is expected to provide a quality-of-service (QoS) guarantee for localization, 
the adopted approximation approach in \cite{Jing_1,Gu_1} will lead to a solution violating the 
accuracy constraints. Therefore, a more advanced optimization strategy is highly recommended to precisely cope with localization CRB in UAV-enabled ISAC networks. In addition, as the sensing targets are very likely not pre-known by UAV, UAV-enabled ISAC networks are also expected to have the capability addressing on-demand sensing services. So far, to the best of our knowledge, a reliable optimization methodology for dealing with QoS-guaranteed localization in UAV-enabled ISAC, as well as the research investigation of on-demand sensing services in UAV-enabled ISAC, is still missing in the literature. 


In this work, we are motivated to consider a UAV-enabled ISAC network with on-demand sensing services, involving on-demand detection and on-demand localization, and aim to develop an effective optimization methodology addressing sensing QoS guarantee in UAV-enabled ISAC. Referring to on-demand sensing, we consider a targeted service region where detection and localization requests may arise toward any position at any time. We apply a UAV with periodic circular trajectory, to supply downlink communication to multiple users and reuse the communication signal for on-demand sensing services. Taking communication fairness into account, the objective is to maximize the minimum throughput among all communication users via jointly designing UAV trajectory and user scheduling scheme, while ensuring satisfactory accuracy levels for the on-demand sensing requests. The main contributions of this work are summarized as following.
\begin{itemize}
    \item {\bf UAV-enabled ISAC with on-demand sensing}: In this work, we have investigated a UAV-enabled ISAC network where the sensing requests are on-demand. 
    For both on-demand detection and on-demand localization, we have characterized the corresponding accuracy constraints respectively as a feasible deployment area for UAV and localization CRB constraints over finite reference sensing targets, to facilitate the overall network design. 
    \item {\bf Localization CRB guarantee in UAV-enabled ISAC}: Towards strict QoS guarantee for localization services, 
    we have devised a new approximation scheme for localization CRB. In particular, our constructed approximation 
    ensures convexity and lower-bound property with respect to the original CRB, such that efficient convex optimization tools can be applied in  the network designs without violating CRB constraints. 
    \item {\bf Iterative algorithm for communication enhancement}: To 
    maximize the throughput performance under these on-demand sensing accuracy constraints, we proposed an iterative algorithm for an efficient suboptimal solution of 
    joint UAV trajectory and user scheduling design. The designed process for constructing convex approximation of the whole problem has ensured a steady convergence of our proposed iterative algorithm.
    \item {\bf Numerical validations}: Via numerical results and comparisons with benchmarks, we particularly validated our developed optimization scheme for effectively dealing with localization CRB constraints and also highlighted the considerable benefits of our proposed solution in communication enhancement. 
\end{itemize}

The remaining sections are 
organized as follows. In Section~II, we state the system model and formulate the optimization problem. In Section~III, we characterize the accuracy constraints for on-demand sensing services and reformulate the problem to a simplified one, which is addressed in Section IV via a proposed iterative algorithm. Finally, the proposed solution is numerically evaluated in Section V and the overall achievements are summarized in Section VI.




\section{System Model}
\begin{figure}[t]
	\centering
	\includegraphics[width=0.6\linewidth, trim= 0 20 0 20]{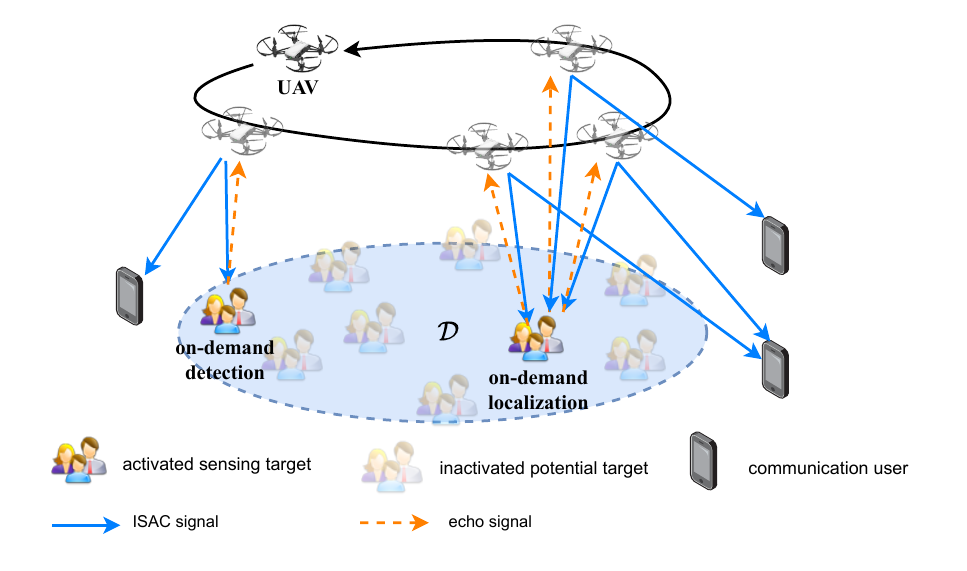}
	\caption{Example of UAV-enabled ISAC network with on-demand sensing.}\label{fig:scenario}
		\vspace{-.5cm}
\end{figure}

In this work, we explore a UAV-enabled ISAC network, where the UAV supports wireless communication to multiple ground users and simultaneously offers on-demand sensing services toward ground targets. 
In other words, 
the sensing tasks in our considered ISAC network may arise at anytime toward any target point in the serving area, as indicated in Fig.~\ref{fig:scenario}. 
As a network example, in post-disaster scenarios where terrestrial base stations are likely offline due to damage, the UAV can serve as a vital temporary base station. 
While transmitting emergency data to multiple ground users, the UAV can leverage ISAC technologies and reuse the emitted RF signal for sensing the disaster area. 
Depending on the real-time needs, 
the on-demand sensing requests at UAV may involve detecting and localizing forest fires, monitoring lava for disaster precaution, and identifying survivals for rescue supports. These on-demand sensing targets are highly likely not antecedently known by the UAV. Beyond post-disaster scenarios, the example applications of our considered network 
also include UAV-enabled monitoring in vehicular networks and autonomous factory, for on-demand detecting and localizing vehicle accidents and robot defects.

We assume the UAV is deployed at fixed altitude $H$ for obstacle avoidance and delivers ISAC services to ground networks following a periodic circular trajectory. Denote by $T$ the operation period for UAV. To facilitate the UAV trajectory design, 
we divide the operation period $T$ into $N$ shorter slots with equal slot length $\delta=\frac{T}{N}$, while in each slot 
the UAV position 
can be assumed to be quasi-static. We represent by $(x_n,y_n)$ the position of UAV in slot $n\in\{1,...,N\}\triangleq \mathcal N$. The UAV speed is limited by a maximum $V$. 
For notation simplicity, we define $(x_{N+n},y_{N+n})=(x_n,y_n)$, $\forall n\in\mathcal N$, such that the UAV speed constraint can be formulated as\vspace{-.1cm}
\begin{equation}\vspace{-.1cm}
	(x_{n+1}-x_n)^2+(y_{n+1}-y_n)^2\leq \delta^2V^2,~\forall n\in\mathcal N.
\end{equation}
Note that when $n=N$, we have $(x_1-x_N)^2+(y_1-y_N)^2\leq \delta^2V^2$, which ensures the UAV trajectory to be circular for continuous wireless service supply. 
\vspace{-.2cm}
\subsection{Wireless Communication Model}

In considered UAV-enabled ISAC networks, we denote by $K$ the total number of ground communication users and by  $\mathbf w_k=(w_{x,k},w_{y,k})$ the ground deployed position for user $k\in\{1,...,K\}\triangleq \mathcal K$. Then, the communication distance between UAV and user $k$ in slot $n$ is given by\vspace{-.1cm}
\begin{equation}\vspace{-.1cm}
	d_n(\mathbf w_k)=\sqrt{(x_n-w_{x,k})^2+(y_n-w_{y,k})^2+H^2}\label{eq:dist_def}.
\end{equation}
Moreover, as pointed out in \cite{uav_1}, the air-to-ground links are highly likely to be dominated by line-of-sight (LoS) paths. Thus, for the wireless channels between UAV and ground users,  we adopt the free-space path loss model as widely applied in \cite{uav_3,uav_4,uav_5} and \cite{Jing_1,free_space1,free_space2}, 
without loss of generality. Given transmit power $P$ at the UAV, we can obtain the channel capacity from UAV to user $k$ in slot $n$ as\vspace{-.1cm}
\begin{equation}\vspace{-.1cm}
	R_{n,k}=\log_2(1+\frac{\beta P}{\sigma^2d_{n}(\mathbf w_k)^2}),
\end{equation}
where $\beta$ denotes the reference channel gain at unit distance and $\sigma^2$ is the noise power. 

Furthermore, we introduce binary variables $a_{n,k}\in\{0,1\}$ for the communication user assignment. Specifically, $a_{k,n}=1$ indicates that the slot $n$ is assigned to user $k$, and vice versa. Hereby, we have the following constraint\vspace{-.1cm}
\begin{equation}\vspace{-.1cm}
	\sum_{k=1}^Ka_{n,k}=1,~\forall n\in\mathcal N,
\end{equation}
indicating that only one user can be served by the UAV in each slot $n$. 
Afterwards, we can formulate the achievable throughput for each user $k$ during an operation period $T$ as
\begin{align}
	U_k(\mathbf x,\mathbf y,\mathbf a)&= \sum_{n=1}^Na_{n,k}\delta R_{k,n}=\sum_{n=1}^Na_{n,k}\delta\log_2(1+\frac{\beta P}{\sigma^2d_{n}(\mathbf w_k)^2}),
\end{align}
where $\mathbf x=[x_1,...,x_N]^\mathsf T\in\mathbb R^{N}$, $\mathbf y=[y_1,...,y_N]^\mathsf T\in\mathbb R^{N}$, and $\mathbf a=[a_{1,1},...,a_{N,K}]^\mathsf T\in\mathbb R^{NK}$.
\vspace{-.1cm}
\subsection{On-Demand Target Detection}

For UAV-enabled sensing services, 
we consider two service types, 
i.e., on-demand target detection and on-demand target localization. Owing to the dynamic nature of wireless networks,   
the on-demand sensing requests may emerge in any slot of the operation period $T$ and for any target point positioned in the serving region. We denote by $\mathcal D$ the targeted region for UAV-aided sensing services, as displayed in Fig.\ref{fig:scenario}, and by $\mathbf s=(s_x,s_y)\in\mathcal D$ an arbitrary target point in the sensing region. 
For target detection, once an on-demand task request occurs, the UAV will immediately inspect the echo of communication signals  from direction of the on-demand target point $\mathbf s\in\mathcal D$. During this process, the receiving  antennas at UAV will operate in a phased-array manner to capture the echo from dedicated direction, enhancing the detection precision. 

Apart from the position-specific feature in on-demand sensing tasks, we assume that the detection request may occur in any slot $n\in\mathcal N$ and will be served timely by the UAV in the same slot. Note that in practice, multiple on-demand tasks may arise within the same slot, which will necessitate a task scheduling process. 
In this paper, 
we assume the on-demand tasks are effectively scheduled by the control centre and informed to the UAV sequentially. In such manner, we are enabled to keep our main research focus on provisioning superior QoS guarantee for on-demand sensing. 

In slot $n$, if a target exists at position $\mathbf s\in\mathcal D$, the echo signal-to-noise ratio (SNR) at the UAV will be given by\vspace{-.1cm}
\begin{equation}\vspace{-.1cm}
	\gamma_{\mathbf s}(x_n,y_n)=\frac{\sigma_r \beta G_rP}{\sigma_0^2 d_n(\mathbf s)^4}, \label{eq:snr_detect}
\end{equation}
where $\sigma_r$ is the factor for radar cross section (RCS), $G_r$ is the receive antenna gain at UAV, and $\sigma_0^2$ denotes the noise power at UAV for sensing. As indicated in \eqref{eq:snr_detect}, the echo SNR is affected by two-way path loss, which results in an inversely proportional property with respect to $d_n(\mathbf s)^4$. Note that to extract echo from the received signal, the UAV will operate in a full-duplex mode and perform self-interference cancellation (SIC). In case of imperfect SIC, there will exist additional interference in the extracted echo signal. Hence, 
we assume the noise power $\sigma_0^2$ is modeled including both thermal noise and residual interference from imperfect SIC. 

To ensure superior on-demand QoS guarantee for target detection, we take into account both position-specific and time-specific features in on-demand detection tasks, which results in the following constraint\vspace{-.1cm}
\begin{equation}\vspace{-.1cm}
	\gamma_{\mathbf s}(x_n,y_n)\geq \xi_d,~\forall n\in\mathcal N,~\forall \mathbf s\in\mathcal D.\label{eq:detect_con}
\end{equation}
The threshold $\xi_d$ is derived out from a maximum allowed 
limit on the detection error probability involving both false alarm probability and miss detection probability. Notice that unlike works \cite{uav_2,uav_3,uav_4,uav_5} considering SNR-based sensing accuracy only for specific targets, the constraint \eqref{eq:detect_con} refers to detection QoS guarantee for the entire service region $\mathcal D$ and for all time slots, presenting additional challenges in the network designs. 
\vspace{-.2cm}
\subsection{On-Demand Target Localization}

Furthermore, in emergency scenarios and real-time monitoring applications, localization acts an important functionality enabling the  precise positioning of accidental events. Accordingly, we also include the capability of UAV for supporting on-demand target localization. Similarly, as on-demand tasks, the localization requests may occur in any slot for any target point $\mathbf s\in\mathcal D$. 
Hereby, different from target detection, the localization tasks 
requires at least three echoes from distinct distances~\cite{sensing_1} while applying the echo time-delay measurement for location estimation. In UAV-enabled ISAC network, we can directly leverage the high mobility of UAV and explore the echoes in different slots, as illustrated in Fig.~\ref{fig:scenario}. The diversely deployed UAV position in different slots will thus enable the localization of ground target points. In particular, we consider $L\geq 3$ consecutive time slots for each target localization request. 
Assuming a target localization request starts from slot $m$ and ends at slot $m+L-1$, for any possible target point $\mathbf s\in\mathcal D$, we have the distance vector denoted by $[d_m(\mathbf s),d_{m+1}(\mathbf s),...,d_{m+L-1}(\mathbf s)]$. As aforementioned, the UAV serves ISAC network periodically via repeating a circular trajectory. In case of $m>N-L+1$, the previously defined $(x_{N+n},y_{N+n})=(x_n,y_n)$, $\forall n\in\mathcal N$, will ensure the following modelling consistent for all $m\in\mathcal N$. 


In the localization process, each distance $d_n(\mathbf s)$ for $m\leq n\leq m+L-1$ can be estimated in slot $n$ via the two-way propagation delay $\tau_n(\mathbf s)$ of received echo at UAV, i.e., $d_n(\mathbf s)=\frac{\tau_n(\mathbf s)c}{2}$ where $c$ is the light speed. The measurement of distance $d_n(\mathbf s)$ can be expressed as\vspace{-.1cm}
\begin{equation}\vspace{-.1cm}
	\hat d_n(\mathbf s)=d_n(\mathbf s) + w_{n,\mathbf s},\label{eq:local_d}
\end{equation}
where $w_{n,\mathbf s}$ represents the Gaussian noise with zero mean and variance $\sigma_{n,\mathbf s}^2$. 
The variances $\sigma_{n,\mathbf s}^2$  reflects the accuracy of distance estimation as well as target localization.  According to \cite{sensing_2}, the variance  $\sigma_{n,\mathbf s}^2$ is inversely proportional to the echo SNR $\gamma_{\mathbf s}(x_n,y_n)$ which is given in \eqref{eq:snr_detect}. Denoting by $\alpha_t$ the scale coefficient, we have \vspace{-.1cm}
\begin{equation}\vspace{-.1cm}
	\sigma_{n,\mathbf s}^2=\frac{\alpha_t}{
		\gamma_{\mathbf s}(x_n,y_n)}.
\end{equation}

For target localization, we are expected to estimate $s_x$ and $s_y$ for any potential target point $\mathbf s\in\mathcal D$. Since MSE of estimators can hardly have an explicit expression, we apply CRB as the localization performance metric~\cite{Gu_1}, which behaves as a lower bound for the MSE of any unbiased estimator. More specifically, denoting by $\hat{\mathbf s}$ the unbiased estimated position for $\mathbf s$, we have the following inequality fulfilled\vspace{-.1cm}
\begin{equation}\vspace{-.1cm}
	\mathbb E_{\mathbf s}\{(\hat{\mathbf s}-\mathbf s)(\hat{\mathbf s}-\mathbf s)^\mathsf{T}\}\geq \mathbf J_{\mathbf s}^{-1}\triangleq \textbf{CRB}_\mathbf s,\label{eq:crb_theory}
\end{equation}
where $\mathbf J_\mathbf s$ is the Fisher information matrix (FIM) of $\mathbf s$. 

Based on \cite{Gu_1} and considering the on-demand localization request initialized in slot $m$, we have the corresponding FIM 
$\mathbf J_{\mathbf s}^{[m]}$ given by\vspace{-.1cm}
\begin{equation}\vspace{-.1cm}
	\mathbf J_{\mathbf s}^{[m]}=\mathbf Q_{\mathbf s}^{[m]} \mathbf J_{\mathbf d}^{[m]}{\mathbf Q_{\mathbf s}^{[m]}}^\mathsf T,\label{eq:crb_step0}
\end{equation}
where $\mathbf Q_{\mathbf s}^{[m]}$ denotes the Jacobian matrix of distance vector $\mathbf d_{\mathbf s}^{[m]}=[d_m(\mathbf s),...,d_{m+L-1}(\mathbf s)]$ with respect to $\mathbf s=(s_x,s_y)$ and $\mathbf J_{\mathbf d}^{[m]}$ is the FIM with respect to the distance vector $\mathbf d_{\mathbf s}^{[m]}$. In particular, we have the Jacobian matrix $\mathbf Q_{\mathbf s}^{[m]}$ given by
\begin{equation}
	\mathbf Q_{\mathbf s}^{[m]}=\frac{\partial {\mathbf d_{\mathbf s}^{[m]}}^\mathsf T}{\partial \mathbf s}=\left[\!
	\begin{array}{cccc}		\frac{x_m-s_x}{d_m(\mathbf s)}&\frac{x_{m+1}-s_x}{d_{m+1}(\mathbf s)}&\cdots&\frac{x_{m+L-1}-s_x}{d_{m+L-1}(\mathbf s)}\\
	\frac{y_m-s_y}{d_m(\mathbf s)}&\frac{y_{m+1}-s_y}{d_{m+1}(\mathbf s)}&\cdots&\frac{y_{m+L-1}-s_y}{d_{m+L-1}(\mathbf s)}
	\end{array}\!
	\right].\label{eq:crb_step1}
\end{equation}
Furthermore, from the distance estimations in \eqref{eq:local_d}, we have the estimated distance vector $\hat{\mathbf d}_{\mathbf s}^{[m]}=[\hat d_m(\mathbf s),...,\hat d_{m+L-1}(\mathbf s)]$ 
following Gaussian distribution with mean $\mathbf d_{\mathbf s}^{[m]}$ and covariance $\mathbf C_{\mathbf s}^{[m]}=\text{diag}(\sigma_{m,\mathbf s}^2,...,\sigma_{m+L-1,\mathbf s}^2)$. Following \cite{sensing_3}, each element $[\mathbf J_{\mathbf d}^{[m]}]_{l_1,l_2}$ in FIM $\mathbf J_{\mathbf d}^{[m]}$ can be obtained as
\begin{align}
	[\mathbf J_{\mathbf d}^{[m]}]_{l_1,l_2}=&\left[\frac{\partial \mathbf d_{\mathbf s}^{[m]}}{\partial d_{l_1}(\mathbf s)} \right]^{\mathsf T}(\mathbf C_{\mathbf s}^{[m]})^{-1}\left[\frac{\partial \mathbf d_{\mathbf s}^{[m]}}{\partial d_{l_2}(\mathbf s)} \right]+\frac{1}{2}\text{tr}\left[(\mathbf C_{\mathbf s}^{[m]})^{-1}\frac{\partial \mathbf C_{\mathbf s}^{[m]}}{\partial d_{l_1}(\mathbf s)}(\mathbf C_{\mathbf s}^{[m]})^{-1}\frac{\partial \mathbf C_{\mathbf s}^{[m]}}{\partial d_{l_2}(\mathbf s)}\right]\!. \label{eq:crb_step2}
\end{align}
By inserting \eqref{eq:crb_step1} and \eqref{eq:crb_step2} into \eqref{eq:crb_step0}, we can finally obtain the FIM $\mathbf J_{\mathbf s}^{[m]}$ with respect to $\mathbf s$, which can be expressed as follows
\begin{equation}
	\mathbf J_{\mathbf s}^{[m]}=\left[
	\begin{aligned}
		\Theta_{a,\mathbf s}^{[m]}&&\Theta_{c,\mathbf s}^{[m]}\\
		\Theta_{c,\mathbf s}^{[m]}&&\Theta_{b,\mathbf s}^{[m]}
	\end{aligned}
	\right],
\end{equation}
where with $\eta\triangleq\frac{\sigma_r\beta G_rP}{\alpha_t\sigma_0^2}$ for notation simplicity, $\Theta_{a,\mathbf s}^{[m]}$, $\Theta_{b,\mathbf s}^{[m]}$ and $\Theta_{c,\mathbf s}^{[m]}$ are respectively given by\vspace{-.1cm}
\begin{equation}
	\Theta_{a,\mathbf s}^{[m]}=\sum_{n=m}^{m+L-1}\left\{\frac{\eta(x_n-s_x)^2}{d_n(\mathbf s)^6}+\frac{8(x_n-s_x)^2}{d_n(\mathbf s)^4}\right\},\label{eq:crb_a}
\end{equation} 
\begin{equation}
	\Theta_{b,\mathbf s}^{[m]}=\sum_{n=m}^{m+L-1}\left\{\frac{\eta(y_n-s_y)^2}{d_n(\mathbf s)^6}+\frac{8(y_n-s_y)^2}{d_n(\mathbf s)^4}\right\},\label{eq:crb_b}
\end{equation}%
\begin{align}
	\Theta_{c,\mathbf s}^{[m]}=\sum_{n=m}^{m+L-1}&\left\{\frac{\eta(x_n-s_x)(y_n-s_y)}{d_n(\mathbf s)^6}\right.\left.+\frac{8(x_n-s_x)(y_n-s_y)}{d_n(\mathbf s)^4}\right\}.\label{eq:crb_c}
\end{align}
Subsequently, we derive out the CRB for the considered localization request at slot $m$ for target point $\mathbf s$, which is given as follows according to \eqref{eq:crb_theory}\vspace{-.1cm}
\begin{align}
	\textbf{CRB}_{\mathbf s}^{[m]}&=(\mathbf J_{\mathbf s}^{[m]})^{-1}
	=\frac{1}{\Theta_{a,\mathbf s}^{[m]}\Theta_{b,\mathbf s}^{[m]}-(\Theta_{c,\mathbf s}^{[m]})^2}\left[
	\begin{aligned}
		\Theta_{b,\mathbf s}^{[m]}&&\Theta_{c,\mathbf s}^{[m]}\\
		\Theta_{c,\mathbf s}^{[m]}&&\Theta_{a,\mathbf s}^{[m]}
	\end{aligned}
	\right].
\end{align}
As the diagonal elements in CRB matrix represents the CRB for each estimated parameter, we have the CRB for the estimation on $\mathbf s$ as follows \cite{sensing_3}\vspace{-.1cm}
\begin{align}
	\Phi_{\mathbf s}^{[m]}(\mathbf x,\mathbf y)&=\text{CRB}_{s_x}+\text{CRB}_{s_y}=\text{Tr}(\textbf{CRB}_\mathbf s)\nonumber \\
	&=\frac{\Theta_{a,\mathbf s}^{[m]}+\Theta_{b,\mathbf s}^{[m]}}{\Theta_{a,\mathbf s}^{[m]}\Theta_{b,\mathbf s}^{[m]}-(\Theta_{c,\mathbf s}^{[m]})^2}.\label{eq:crb}
\end{align}

Hereby, we define a maximum allowable CRB, $\xi_l$, to guarantee a superior QoS level in the on-demand target localization services provided by the UAV. 
In other words, we have the following constraint
\begin{equation}
	\Phi_{\mathbf s}^{[m]}(\mathbf x,\mathbf y)\leq \xi_l,~\forall m\in\mathcal N,~\forall \mathbf s\in\mathcal D,\label{eq:local_con}
\end{equation}
which ensures satisfactory QoS for on-demand localization requests initialized in any slot and for any target point in serving region $\mathcal D$. 
Note that the expression of CRB $\Phi_{\mathbf s}^{[m]}(\mathbf x,\mathbf y)$ given in \eqref{eq:crb} is highly complicated and  nonconvex, 
introducing significant difficulties in the corresponding system designs.
As aforedmentioend, the popular approximation approach applied in existing works \cite{Gu_1,Jing_1} for dealing with CRB in localization can hardly be adopted for scenarios towards localization QoS guarantee. The straightforward first-order Taylor approximation leads to considerable approximation errors. Moreover, 
the infinite number of constraints in \eqref{eq:local_con}, corresponding to the infinite number of target points in serving region $\mathcal D$, makes the on-demand QoS guarantee more challenging to achieve in the overall network designs.

\subsection{Problem Formulation}

In this work, we aim at maximally enhancing the communication performance via the joint communication user assignment and UAV trajectory design, while guaranteeing on-demand sensing QoS for both on-demand detection and on-demand localization tasks. To achieve a fairness among multiple communication users, we apply the minimum achievable throughput among multiple users as the objective to be maximized. The original optimization problem can be formulated as\vspace{-.1cm}
\begin{subequations}
	\begin{align}
		(\text{OP}): \max_{\mathbf x,\mathbf y,\mathbf a} &~~ \min_{k\in\mathcal K}\{U_k(\mathbf x,\mathbf y,\mathbf a)\}\label{eq:op_obj} \\
		\mathrm{s.t.}~&~~(x_{n+1}\!-\!x_{n})^2\!+\!(y_{n+1}\!-\!y_{n})^2\!\leq\! \delta^2V^2,\forall n,\\
		&~~\gamma_{\mathbf s}(x_n,y_n)\geq \xi_d,~\forall n\in\mathcal N,~\forall \mathbf s\in\mathcal D,\label{eq:op_con2}\\
		&~~\Phi_{\mathbf s}^{[m]}(\mathbf x,\mathbf y)\leq \xi_l,~\forall m\in\mathcal N,~\forall \mathbf s\in\mathcal D,\label{eq:op_con3}\\
		&~~ \sum_{k=1}^Ka_{k,n}=1,~\forall n\in\mathcal N,\\
		&~~ a_{k,n}\in\{0,1\},~\forall k\in\mathcal K,~\forall n\in\mathcal N.
	\end{align}
\end{subequations}
Clearly, the problem (OP) is 
not convex, due to the nonconcavity of objective \eqref{eq:op_obj} and the nonconvexity of constraints \eqref{eq:op_con2} and \eqref{eq:op_con3}. In the meantime, the binary user assignment variable $\mathbf a$, the aforementioned infinite number of constraints in \eqref{eq:op_con2} and \eqref{eq:op_con3} for guaranteeing on-demand sensing QOS, as well as the highly sophisticated CRB expression in \eqref{eq:op_con3}, have made the problem (OP) particularly strenuous to be tackled.


\section{Characterization of On-Demand Sensing Requirements}

To address the above design difficulties, we start with characterizing the impacts of infinite on-demand sensing constraints \eqref{eq:op_con2} and \eqref{eq:op_con3}. Based on the characterization, we will identify a feasible region for UAV deployment and finite number of reference localization constraints, which will largely facilitate the problem analysis.
\vspace{-.2cm}
\subsection{UAV Deployment Restriction for On-Demand Detection} 
Referring to the on-demand detection QoS guarantee, the constraint \eqref{eq:op_con2} poses a minimum SNR limit to any wireless connection from a deployed UAV position $(x_n,y_n)$ to a target point $\mathbf s$ in ground region $\mathcal D$. In other words, for each target point $\mathbf s\in\mathcal D$, the echo SNR in all UAV operation slots should be above the minimum SNR limit, which restricts the whole trajectory of UAV in a certain region. More specifically, for given target point $\mathbf s\in\mathcal D$, the constraint \eqref{eq:op_con2} is equivalent to a limitation on the distance from the UAV to $\mathbf s$, i.e.,\vspace{-.1cm}
\begin{equation}\vspace{-.1cm}
	d_n(\mathbf s)\leq \sqrt[4]{\frac{\sigma_r\beta G_rP}{\sigma_0^2\xi_d}},~\forall n\in\mathcal N.
\end{equation}
Further considering the fixed deployment altitude $H$ for the UAV, we have an equivalent limitation on the horizontal distance $\bar d_n(\mathbf s)$ between UAV and point $\mathbf s$\vspace{-.1cm}
\begin{align}
	\bar d_n(\mathbf s)&=\sqrt{(x_n-s_x)^2+(y_n-s_y)^2}=\sqrt{d_n(\mathbf s)^2-H^2}\nonumber \\
	&\leq \Big(\sqrt{\frac{\sigma_r\beta G_rP}{\sigma_0^2\xi_d}}-H^2\Big)^{1/2}{\triangleq \bar d_0},~\forall n\in\mathcal N,
\end{align}
restricting the whole UAV trajectory in a circular area with centre $(s_x,s_y)$ and radius $\bar d_0$. 

Afterwards, for guaranteeing detection QoS for the whole serving region $\mathcal D$, we can obtain the following Lemma~\ref{le:detect_region}, clearly stating the feasible UAV deployment region corresponding to the constraint \eqref{eq:op_con2}.\vspace{-.1cm}
\begin{lemma}\label{le:detect_region}
	In problem (OP), the constraint \eqref{eq:op_con2} for  on-demand detection QoS guarantee, is equivalent to a constraint restricting the whole UAV trajectory in a convex region $\Omega$,
	which is given by\vspace{-.1cm}
	\begin{equation}\vspace{-.1cm}
		\Omega=\bigcap_{\mathbf s\in\mathcal D}\{(x,y)|\sqrt{(x-s_x)^2+(y-s_y)^2}\leq r_0\}.\label{eq:le_region}
	\end{equation}
\end{lemma}
\begin{proof}
	As aforementioned, guaranteeing detection QoS for each target point $\mathbf s$ will constrict the UAV trajectory in a circular area. For QoS guarantee with respect to the whole serving region, the UAV trajectory is thus restricted in the intersection of these circular areas, as shown in \eqref{eq:le_region}. Since all circular areas are convex, the intersection region $\Omega$ can be easily proved to be convex. 
\end{proof}
According to Lemma~\ref{le:detect_region}, the constraint \eqref{eq:op_con2} can be equivalently reformed as\vspace{-.1cm}
\begin{equation}\vspace{-.1cm}
	(x_n,y_n)\in\Omega,
	~\forall n\in\mathcal N.\label{eq:op_con2_new}
\end{equation}
The convexity of region $\Omega$
shown in Lemma~\ref{le:detect_region} ensures the convexity of the reformed constraint~\eqref{eq:op_con2_new}.

\subsection{Reference Target Points for On-Demand Localization}

For on-demand localization services, to deal with the infinite number of potential sensing points in region $\mathcal D$, we propose an approach for defining finite reference target points. 
As a consequence, satisfying the on-demand localization constraints at these reference target points will guarantee the constraint fulfillment over the whole serving region $\mathcal D$, with sufficiently high probability. 

Specifically, we start with a set $\Gamma_l$ of randomly generated sensing target points  within $\mathcal D$, a set $\Gamma_t$ of randomly generated UAV trajectory satisfying all constraints in (OP) except \eqref{eq:op_con3}, and a set $\mathcal F=\emptyset$ for saving the reference target points. For the trajectories in set $\Gamma_t$, we first delete these which cannot satisfy \eqref{eq:op_con3} for all $\mathbf s\in\mathcal F$.  For each $\mathbf s_i\in\Gamma_l$, if there exists any trajectory in reduced set $\Gamma_t$ not satisfying \eqref{eq:op_con3}, then this target point $\mathbf s_i$ should be identified as a new reference point, as there exists a feasible trajectory passing sensing tests \eqref{eq:op_con3} at all sensing points in $\mathcal F$ but still not fulfilling the constraint \eqref{eq:op_con3} at $\mathbf s_i$. We then add $\mathbf s_i$ into set $\mathcal F$, to enrich the set of reference target points. 

After each round update of $\mathcal F$, we regenerate new sets $\Gamma_l$ and $\Gamma_t$, for the next round $\mathcal F$ updates. Finally, if the size of set $\mathcal F$ is not enlarged after many-round updates, we are then convinced that the set $\mathcal F$ has contained sufficient reference target points, such that fulfilling the sensing constraint \eqref{eq:op_con3} at these points will sufficiently represent a fulfillment of \eqref{eq:op_con3} for all sensing targets in serving region $\mathcal D$. 

\subsection{Problem Reformulation}
Applying the above characterizations on constraints \eqref{eq:op_con2} and \eqref{eq:op_con3}, as well as a relaxation on the binary assignment variable $\mathbf a$, we obtain a reformulated problem of (OP) as\vspace{-.1cm}
\begin{subequations}
	\begin{align}
		(\text{P1}): \max_{\mathbf x,\mathbf y,\mathbf a} &~~ \min_{k\in\mathcal K}\{U_k(\mathbf x,\mathbf y,\mathbf a)\}\label{eq:p1_obj} \\
		\mathrm{s.t.}~&~~(x_{n+1}\!-\!x_{n})^2\!+\!(y_{n+1}\!-\!y_{n})^2\!\leq\! \delta^2V^2,\forall n,\label{eq:p1_con1}\\
		&~~(x_n,y_n)\in\Omega,~\forall n\in\mathcal N,\label{eq:p1_con2}\\
		&~~\Phi_{\mathbf s}^{[m]}(\mathbf x,\mathbf y)\leq \xi_l,~\forall m\in\mathcal N,~\forall \mathbf s\in\mathcal F,\label{eq:p1_con3}\\
		&~~ \sum_{k=1}^Ka_{k,n}=1,~\forall n\in\mathcal N,\\
		&~~ 0\leq a_{k,n}\leq 1,~\forall k\in\mathcal K,~\forall n\in\mathcal N.
	\end{align}
\end{subequations}
With relaxation on  binary variable $\mathbf a$, the resulting solution from solving (P1) may contain non-integer assignment decision. However, according to \cite{our_1}, we have the conclusion that an assignment solution along UAV trajectory, where multiple users are scheduled to share trajectory segment with 
different occupation ratios, can be practically realized without performance degradation via an overall user rescheduling. 
This implies that compared to (OP), the reformulated problem (P1) with relaxed  assignment variables is without loss of design optimality. In the following, we will concentrate on the still nonconvex problem (P1) and target at proposing an efficient solution. 
\vspace{-.2cm}
\section{Iterative Solution}
\vspace{-.1cm}
Via observations, we can find that the problem (P1) can be approximated to a convex one, if we can build up a concave approximation for the objective \eqref{eq:p1_obj} and a convex approximation for the constraint \eqref{eq:p1_con3}. In this section, 
we make efforts to establish a convex problem approximation for (P1) with zero approximation gap at given feasible point $(\mathbf x^{(r)},\mathbf y^{(r)},\mathbf a^{(r)})$. The approximation approach assures that solving the approximated convex problem will lead to an improved feasible solution, thus enabling an iterative algorithm for (P1) via successive convex programming (SCP) techniques.\vspace{-.3cm}
\subsection{Objective Approximation}

For the nonconcave objective to be maximized, we aim to construct a concave approximation as the lower bound function, such that maximizing the approximated function will guarantee improvement of the objective function. As aforementioned, the approximation will be tight at given local point $(\mathbf x^{(r)},\mathbf y^{(r)},\mathbf a^{(r)})$, to enable iterative solution improvements. 

Hereby, we start with the approximation construction for the accumulated throughput at user $k$ during slot $n$, namely, $a_{n,k}\delta\log_2(1+\frac{\beta P}{\sigma^2d_{n}(\mathbf w_k)^2})$. Base on the easily proved convexity of function $\log_2(1+\frac{1}{t})$ with respect to $t>0$, we have the following inequality\vspace{-.2cm}
\begin{align}
	&a_{n,k}\delta \log_2(1+\frac{\beta P}{\sigma^2d_{n}(\mathbf w_k)^2})
	\geq-A_{1,k,n}^{(r)}a_{n,k}d_n(\mathbf w_k)^2+A_{2,n,k}^{(r)}a_{n,k},\label{eq:obj_app1}
\end{align}
where the constants $A_{1,k,n}^{(r)}$ and $A_{2,k,n}^{(r)}$ are defined at the local point $(\mathbf x^{(r)},\mathbf y^{(r)},\mathbf a^{(r)})$ as\vspace{-.1cm}
\begin{align}
	A_{1,k,n}^{(r)}&=\frac{\delta\frac{\beta P}{\sigma^2}\ln2}{d_n^{(r)}(\mathbf w_k)^2(d_n^{(r)}(\mathbf w_k)^2+\frac{\beta P}{\sigma^2})}>0,\\
	A_{2,k,n}^{(r)}&=\delta \log_2(1+\frac{\beta P}{\sigma^2d_{n}^{(r)}(\mathbf w_k)^2})+A_{1,k,n}^{(r)}d_{n}^{(r)}(\mathbf w_k)^2,
\end{align}
ensuring the equality in \eqref{eq:obj_app1} holds at $(\mathbf x^{(r)},\mathbf y^{(r)},\mathbf a^{(r)})$. Note that $d_{n}^{(r)}(\mathbf w_k)$ is constant with given $(\mathbf x^{(r)},\mathbf y^{(r)},\mathbf a^{(r)})$ and denotes the value of $d_{n}(\mathbf w_k)$ at the local point. 

Since the approximated result in \eqref{eq:obj_app1} is still not concave, we are motivated to apply a further step approximation based on the mean inequality. As known, we have $\forall a,b>0$, $ab\leq \frac{1}{2}a^2+\frac{1}{2}b^2$ and the equality holds when $a=b$. Thus, we can introduce a positive constant 
	$A_{3,k,n}^{(r)}=\frac{a_{k,n}^{(r)}}{d_n^{(r)}(\mathbf w_k)^2}$
and treat $a_{n,k}$ and $A_{3,k,n}^{(r)}d_n(\mathbf w_k)^2$ in \eqref{eq:obj_app1} respectively as $a$ and $b$ in the mean inequality. Accordingly, we have\vspace{-.1cm}
\begin{align}
	a_{n,k}\delta \log_2(1+\frac{\beta P}{\sigma^2d_{n}(\mathbf w_k)^2})&
	\geq-\frac{A_{1,k,n}^{(r)}}{2A_{3,k,n}^{(r)}}a_{n,k}^2-\frac{A_{1,k,n}^{(r)}A_{3,k,n}^{(r)}}{2}d_n(\mathbf w_k)^4+A_{2,n,k}^{(r)}a_{n,k}\nonumber \\
	&\triangleq f_{n,k}^{(r)}(\mathbf x,\mathbf y,\mathbf a),
\end{align}
where aided by $A_{3,k,n}^{(r)}$, the equality is guaranteed to hold at the local point  $(\mathbf x^{(r)},\mathbf y^{(r)},\mathbf a^{(r)})$. The obtained function $f_{n,k}^{(r)}(\mathbf x,\mathbf y,\mathbf a)$ can be easily proved to be concave. 

After all, we can accomplish the approximation for the whole objective function as\vspace{-.2cm}
\begin{equation}\vspace{-.15cm}
	\min_{k\in\mathcal K}\{U_k(\mathbf x,\mathbf y,\mathbf a)\}\geq \min_{k\in\mathcal K}\{\sum_{n=1}^Nf_{n,k}^{(r)}(\mathbf x,\mathbf y,\mathbf a)\},
\end{equation}
which is clearly concave and shows a tight approximation at the local point $(\mathbf x^{(r)},\mathbf y^{(r)},\mathbf a^{(r)})$.
\vspace{-.2cm}
\subsection{Approximation of On-Demand Localization Constraint}
Next, we move forward to the approximation of the nonconvex on-demand localization constraint \eqref{eq:p1_con3}. To facilitate the approximation process, we reform the constraint \eqref{eq:p1_con3} equivalently as follows $\forall m\in\mathcal N$ and $\forall \mathbf s\in\mathcal F$\vspace{-.1cm}
\begin{equation}\vspace{-.1cm}
	\Theta_{a,\mathbf s}^{[m]}+\Theta_{b,\mathbf s}^{[m]}- \xi_l(\Theta_{a,\mathbf s}^{[m]}\Theta_{b,\mathbf s}^{[m]}-(\Theta_{c,\mathbf s}^{[m]})^2)\leq 0.\label{eq:crb_app0}
\end{equation}
Recalling the detailed expressions in \eqref{eq:crb_a}-\eqref{eq:crb_c}, there exists dramatic difficulties in dealing with  the constraint \eqref{eq:crb_app0}, 
making it particularly challenging to construct a convex approximation for the constraint function in \eqref{eq:crb_app0} while ensuring both the upper bound property and approximation tightness at given local point. 
To address these challenges, we will focus on the convex upper bound approximations for $\Theta_{a,\mathbf s}^{[m]}+\Theta_{b,\mathbf s}^{[m]}$, and the concave lower bound approximation for $\Theta_{a,\mathbf s}^{[m]}\Theta_{b,\mathbf s}^{[m]}-(\Theta_{c,\mathbf s}^{[m]})^2$, respectively. 
\subsubsection{Convex Approximations for $\Theta_{a,\mathbf s}^{[m]}+\Theta_{b,\mathbf s}^{[m]}$}
According to the definition of $d_n(\mathbf s)$ in~\eqref{eq:dist_def},
the summed term $\Theta_{a,\mathbf s}^{[m]}+\Theta_{b,\mathbf s}^{[m]}$ can be reformulated as\vspace{-.1cm} 
\begin{align}
    \Theta_{a,\mathbf s}^{[m]}\!+\!\Theta_{b,\mathbf s}^{[m]}\!=\!\!\!\sum_{n=m}^{m+L-1}\!\!\left\{\!\frac{\eta}{d_n(\mathbf s)^4}-\frac{\eta H^2}{d_n(\mathbf s)^6}+\frac{8}{d_n(\mathbf s)^2}-\frac{8H^2}{d_n(\mathbf s)^4}\!\right\}\!\!.
\end{align}
In the above expression, it can be easily proved that the term $-\frac{\eta H^2}{d_n(\mathbf s)^6}-\frac{8H^2}{d_n(\mathbf s)^4}$ is concave in $d_n(\mathbf s)^2$. Leveraging the first-order condition of concave functions, we have\vspace{-.1cm}
\begin{align}
    &\Theta_{a,\mathbf s}^{[m]}+\Theta_{b,\mathbf s}^{[m]}
    \leq\!\sum_{n=m}^{m+L-1}\!\left\{\!\frac{\eta}{d_n(\mathbf s)^4}+\frac{8}{d_n(\mathbf s)^2}+B_{1,n,\mathbf s}^{(r)}{d_n(\mathbf s)^2}\!+B_{2,n,\mathbf s}^{(r)}\!\right\}\!,\label{eq:ab_app1}
\end{align}
where constants $B_{1,n,\mathbf s}^{(r)}$ and $B_{2,n,\mathbf s}^{(r)}$ are respectively given by\vspace{-.1cm}
\begin{align}
    B_{1,n,\mathbf s}^{(r)}&=\frac{3\eta H^2}{(d_n^{(r)}(\mathbf s))^8}+\frac{16H^2}{(d_n^{(r)}(\mathbf s))^6}>0,\\
    B_{2,n,\mathbf s}^{(r)}&=-\frac{\eta H^2}{(d_n^{(r)}(\mathbf s))^6}-\frac{8H^2}{(d_n^{(r)}(\mathbf s))^4}-B_{1,n,\mathbf s}^{(r)}(d_n^{(r)}(\mathbf s))^2,
\end{align}
with $d_n^{(r)}(\mathbf s)$ representing the value of $d_n(\mathbf s)$ at local point $(\mathbf x^{(r)},\mathbf y^{(r)},\mathbf a^{(r)})$. The equality in approximation \eqref{eq:ab_app1} can be verified holding at the given local point $(\mathbf x^{(r)},\mathbf y^{(r)},\mathbf a^{(r)})$. 
Since $B_{1,n,\mathbf s}^{(r)}>0$ in \eqref{eq:ab_app1}, we can accomplish the convex approximation for $\Theta_{a,\mathbf s}^{[m]}+\Theta_{b,\mathbf s}^{[m]}$ via constructing a tight convex upper bound for the nonconvex term $\frac{\eta}{d_n(\mathbf s)^4}+\frac{8}{d_n(\mathbf s)^2}$. Regarding the convex term $d_n(\mathbf s)^2$, we have \vspace{-.1cm}
\begin{align}
    d_n(\mathbf s)^2\geq& 2(x_n^{(r)}-s_x)(x_n-s_x)+2(y_n^{(r)}-s_y)(y_n-s_y)+2H^2-(d_n^{(r)}(\mathbf s))^2\nonumber \\
    \triangleq&\varphi_{n,\mathbf s}^{(r)}(x_n,y_n),
\end{align}
where the equality holds at point $(\mathbf x^{(r)},\mathbf y^{(r)},\mathbf a^{(r)})$. Further applying the monotonic decreasing property of the term $\frac{\eta}{d_n(\mathbf s)^4}+\frac{8}{d_n(\mathbf s)^2}$ with respect to $d_n(\mathbf s)^2$, we has the convex approximation for $\Theta_{a,\mathbf s}^{[m]}+\Theta_{b,\mathbf s}^{[m]}$ given as
\begin{align}\vspace{-.1cm}
    &\Theta_{a,\mathbf s}^{[m]}+\Theta_{b,\mathbf s}^{[m]}\nonumber \\
    \leq&\!\sum_{n=m}^{m+L-1}\!\Big\{\!\frac{\eta}{(\varphi_{n,\mathbf s}^{(r)}(x_n,y_n))^2}+\frac{8}{\varphi_{n,\mathbf s}^{(r)}(x_n,y_n)}+B_{1,n,\mathbf s}^{(r)}{d_n(\mathbf s)^2}+B_{2,n,\mathbf s}^{(r)}\!\Big\}\!\nonumber \\
    \triangleq &g_{1,m,\mathbf s}^{(r)}(\mathbf x,\mathbf y),\label{eq:ab_app2}
\end{align}
where the equality is also guaranteed to hold at the local point $(\mathbf x^{(r)},\mathbf y^{(r)},\mathbf a^{(r)})$. Hereby, it can be proved that the obtained function $g_{1,m,\mathbf s}^{(r)}(\mathbf x,\mathbf y)$ is jointly convex in $(\mathbf x,\mathbf y)$.

\subsubsection{Concave Approximation for $\Theta_{a,\mathbf s}^{[m]}\Theta_{b,\mathbf s}^{[m]}-(\Theta_{c,\mathbf s}^{[m]})^2$}
Next, we turn to construct a concave lower bound function for $\Theta_{a,\mathbf s}^{[m]}\Theta_{b,\mathbf s}^{[m]}-(\Theta_{c,\mathbf s}^{[m]})^2$. According to \cite{ine_0}, we can rewrite the term  $\Theta_{a,\mathbf s}^{[m]}\Theta_{b,\mathbf s}^{[m]}-(\Theta_{c,\mathbf s}^{[m]})^2$ as\vspace{-.1cm}
\begin{align}
	&\Theta_{a,\mathbf s}^{[m]}\Theta_{b,\mathbf s}^{[m]}-(\Theta_{c,\mathbf s}^{[m]})^2\nonumber \\
	=&\sum_{n_1=m}^{m+L-2}\sum_{n_2=n_1+1}^{m+L-1}\Big\{\big(\frac{\eta^2}{d_{n_1}(\mathbf s)^6d_{n_2}(\mathbf s)^6}+\frac{8\eta}{d_{n_1}(\mathbf s)^6d_{n_2}(\mathbf s)^4}+\frac{8\eta}{d_{n_1}(\mathbf s)^4d_{n_2}(\mathbf s)^6}+\frac{64}{d_{n_1}(\mathbf s)^4d_{n_2}(\mathbf s)^4}\big)\nonumber \\
	&\boldsymbol{\cdot}\underbrace{\big[\!(x_{n_1}-s_x)(y_{n_2}-s_y)-(x_{n_2}-s_x)(y_{n_1}-s_y)\!\big]^2\!}_{\triangleq \psi_{\mathbf s,n_1,n_2}(\mathbf x,\mathbf y)}\Big\},\!\label{eq:crb_c_new}
\end{align}
which will significantly facilitate the construction of targeted concave approximation. Specifically, we have the function $h_2(d_1^2,d_2^2)=\frac{\eta^2}{d_1^6d_2^6}+\frac{8\eta}{d_1^6d_2^4}+\frac{8\eta}{d_1^4d_2^6}+\frac{64}{d_1^4d_2^4}$ being jointly convex in $(d_1^2,d_2^2)$ and monotonically decreasing in both $d_1^2$ and $d_2^2$. Applying the property of convex function $h_2(d_1^2,d_2^2)$ in \eqref{eq:crb_c_new}, we can obtain\vspace{-.1cm}
\begin{align}
	&\big(\frac{\eta^2}{d_{n_1}^6d_{n_2}^6}+\frac{8\eta}{d_{n_1}^6d_{n_2}^4}+\frac{8\eta}{d_{n_1}^4d_{n_2}^6}+\frac{64}{d_{n_1}^4d_{n_2}^4}\big)\cdot \psi_{\mathbf s,n_1,n_2}(\mathbf x,\mathbf y)\nonumber \\
	\geq &-\psi_{\mathbf s,n_1,n_2}(\mathbf x,\mathbf y)\big({E}_{1,n_1,n_2}^{(r)}d_{n_1}^2+{E}_{2,n_1,n_2}^{(r)}d_{n_2}^2\big)+{E}_{3,n_1,n_2}^{(r)}\psi_{\mathbf s,n_1,n_2}(\mathbf x,\mathbf y),\label{eq:crb_c_app1}
\end{align}
where we define $d_{n_1}=d_{n_1}(\mathbf s)$ and $d_{n_2}=d_{n_2}(\mathbf s)$ for denotation simplicity. The constants ${E}_{1,n_1,n_2}^{(r)}$, ${E}_{2,n_1,n_2}^{(r)}$ and ${E}_{3,n_1,n_2}^{(r)}$ are all non-negative and associated to the local point $(\mathbf x^{(r)},\mathbf y^{(r)},\mathbf a^{(r)})$. In particular, these non-negative constants are defined as follows, according to the first-order condition of convex functions\vspace{-.1cm}
\begin{align}
    {E}_{1,n_1,n_2}^{(r)}&\!=\!\frac{3\eta^2}{{d_{n_1}^{(r)}}^8\!{d_{n_2}^{(r)}}^6}\!+\!\frac{24\eta}{{d_{n_1}^{(r)}}^8\!{d_{n_2}^{(r)}}^4}\!+\!\frac{16\eta}{{d_{n_1}^{(r)}}^6\!{d_{n_2}^{(r)}}^6}\!+\!\frac{128}{{d_{n_1}^{(r)}}^6\!{d_{n_2}^{(r)}}^4},\\
    {E}_{2,n_1,n_2}^{(r)}&\!=\!\frac{3\eta^2}{{d_{n_1}^{(r)}}^6\!{d_{n_2}^{(r)}}^8}\!+\!\frac{16\eta}{{d_{n_1}^{(r)}}^6\!{d_{n_2}^{(r)}}^6}\!+\!\frac{24\eta}{{d_{n_1}^{(r)}}^4\!{d_{n_2}^{(r)}}^8}\!+\!\frac{128}{{d_{n_1}^{(r)}}^4\!{d_{n_2}^{(r)}}^6},\\
    {E}_{3,n_1,n_2}^{(r)}&\!=\!h_2({d_{n_1}^{(r)}}^2,{d_{n_2}^{(r)}}^2)+{E}_{1,n_1,n_2}^{(r)}{d_{n_1}^{(r)}}^2+{E}_{1,n_1,n_2}^{(r)}{d_{n_1}^{(r)}}^2,
\end{align}

\vspace{-.2cm}
\noindent{}where ${d_{n_1}^{(r)}}$ and ${d_{n_2}^{(r)}}$ are respectively the value of ${d_{n_1}}$ and ${d_{n_2}}$ at the local point $(\mathbf x^{(r)},\mathbf y^{(r)},\mathbf a^{(r)})$. Clearly, it can be easily verified that the equality in \eqref{eq:crb_c_app1} holds at the local point $(\mathbf x^{(r)},\mathbf y^{(r)},\mathbf a^{(r)})$. 

Next, on account of the complicated expression of function $\psi_{\mathbf s,n_1,n_2}(\mathbf x,\mathbf y)$ in \eqref{eq:crb_c_app1}, we move forward to decouple  $\psi_{\mathbf s,n_1,n_2}(\mathbf x,\mathbf y)$ from other terms in the next-step approximation. Recall a variant of the mean inequality, i.e., $ab=\frac{1}{E}\cdot a\cdot Eb\leq \frac{1}{2E}a^2+\frac{E}{2}b^2$, $\forall a,b,E\geq 0$, where equality holds when $a=Eb$. By letting $a=\psi_{\mathbf s,n_1,n_2}(\mathbf x,\mathbf y)\geq0$ and $b={E}_{1,n_1,n_2}^{(r)}d_{n_1}^2+{E}_{2,n_1,n_2}^{(r)}d_{n_2}^2\geq 0$, the result in \eqref{eq:crb_c_app1} can be further approximated as\vspace{-.1cm}
\begin{align}
	&\big(\frac{\eta^2}{d_{n_1}^6d_{n_2}^6}+\frac{8\eta}{d_{n_1}^6d_{n_2}^4}+\frac{8\eta}{d_{n_1}^4d_{n_2}^6}+\frac{64}{d_{n_1}^4d_{n_2}^4}\big)\cdot \psi_{\mathbf s,n_1,n_2}(\mathbf x,\mathbf y)\nonumber \\
	\geq &-\frac{{E}_{4,n_1,n_2}^{(r)}}{2}\big({E}_{1,n_1,n_2}^{(r)}d_{n_1}^2+{E}_{2,n_1,n_2}^{(r)}d_{n_2}^2\big)^2-\frac{1}{2{E}_{4,n_1,n_2}^{(r)}\!\!}\psi_{\mathbf s,n_1,n_2}(\mathbf x,\mathbf y)^2\!+\!E_{3,n_1,n_2}^{(r)}\psi_{\mathbf s,n_1,n_2}(\mathbf x,\mathbf y),\label{eq:crb_c_app2}
\end{align}

\vspace{-.2cm}
\noindent{}where the constant ${E}_{4,n_1,n_2}^{(r)}$ are defined as\vspace{-.1cm}
\begin{equation}\vspace{-.1cm}
    {E}_{4,n_1,n_2}^{(r)}=\frac{\psi_{\mathbf s,n_1,n_2}(\mathbf x^{(r)},\mathbf y^{(r)})}{{E}_{1,n_1,n_2}^{(r)}{d_{n_1}^{(r)}}^2+{E}_{2,n_1,n_2}^{(r)}{d_{n_2}^{(r)}}^2},
\end{equation}
ensuring the equality in  approximation \eqref{eq:crb_c_app2} holds at local point  $(\mathbf x^{(r)},\mathbf y^{(r)},\mathbf a^{(r)})$. So far, in  \eqref{eq:crb_c_app2}, only the terms related to complicated function $\psi_{\mathbf s,n_1,n_2}(\mathbf x,\mathbf y)$ are not concave yet. To achieve the targeted concave approximation, we propose the following Lemma~\ref{le:crb_app2}, to aid the last-step approximation. 
\vspace{-.1cm}
\begin{lemma}\label{le:crb_app2}
    Consider function $h_3(\mathbf x_0,\mathbf y_0)=(x_1y_2-x_2y_1)^2$ with $\mathbf x_0=[x_1,x_2]^\mathsf T$ and $\mathbf y_0=[y_1,y_2]^\mathsf T$. Then, given local point $(\mathbf x_0^{(r)},\mathbf y_0^{(r)})=(x_1^{(r)},x_2^{(r)},y_1^{(r)},y_2^{(r)})$, we have the following  inequalities holds\vspace{-.1cm}
	\begin{align}
        &\Big(4h_4(\mathbf x_0^{(r)},\mathbf y_0^{(r)})\tilde{h}_4^{(r)}(\mathbf x_0,\mathbf y_0)+4h_5(\mathbf x_0^{(r)},\mathbf y_0^{(r)})\tilde{h}_5^{(r)}(\mathbf x_0,\mathbf y_0)-(h_4(\mathbf x_0,\mathbf y_0)+h_5(\mathbf x_0,\mathbf y_0))^2\Big)\nonumber \\
        \leq& 
        h_3(\mathbf x_0,\mathbf y_0)
        \nonumber \\
        \leq & 
        \Big(\!\!\max\!\big\{h_4(\mathbf x_0,\mathbf y_0)\!-\!\tilde{h}_5^{(r)}\!(\mathbf x_0,\mathbf y_0),
        h_5(\mathbf x_0,\mathbf y_0)\!-\!\tilde{h}_4^{(r)}\!(\mathbf x_0,\mathbf y_0)\big\}\!\Big)^{\!2}\!\!,\label{eq:le_app3}
	\end{align}
	
    \vspace{-.2cm}
\noindent{}where $h_4(\mathbf x_0,\mathbf y_0)$ and $h_5(\mathbf x_0,\mathbf y_0)$ are positive convex functions defined as\vspace{-.1cm}
    \begin{align}
        h_4(\mathbf x_0,\mathbf y_0)&=\frac{1}{2}\big((x_1+y_2)^2+x_2^2+y_1^2\big),
        \end{align}
        \begin{align}
        h_5(\mathbf x_0,\mathbf y_0)&=\frac{1}{2}\big(x_1^2+y_2^2+(x_2+y_1)^2\big),
    \end{align}
    and $\tilde{h}_4^{(r)}(\mathbf x_0,\mathbf y_0)$, $\tilde{h}_5^{(r)}(\mathbf x_0,\mathbf y_0)$ are respectively the first-order Taylor expansions of functions $h_4(\mathbf x_0,\mathbf y_0)$ and $h_5(\mathbf x_0,\mathbf y_0)$ at the local point $(\mathbf x_0^{(r)},\mathbf y_0^{(r)})$, i.e.,\vspace{-0.1cm}
    \begin{align}
        \tilde{h}_4^{(r)}(\mathbf x_0,\mathbf y_0)=&(x_1^{(r)}+y_2^{(r)})(x_1+y_2)+x_2^{(r)}x_2+y_1^{(r)}y_1-{h}_4(\mathbf x_0^{(r)},\mathbf y_0^{(r)}),\\
        \tilde{h}_5^{(r)}(\mathbf x_0,\mathbf y_0)=&x_1^{(r)}x_1+y_2^{(r)}y_2+(x_2^{(r)}+y_1^{(r)})(x_2+y_1)-{h}_5(\mathbf x_0^{(r)},\mathbf y_0^{(r)}).
    \end{align}
    
    \vspace{-0.2cm}
	In particular, we have both  equalities in \eqref{eq:le_app3} hold at the local point $(x_1^{(r)},x_2^{(r)},y_1^{(r)},y_2^{(r)})$. The left-hand side function in \eqref{eq:le_app3} is concave, while the right-hand side of \eqref{eq:le_app3} is a convex function. 
\end{lemma}
\begin{proof}
    To start with, we first focus on the proof of left-hand side inequality on $h_3(\mathbf x_0,\mathbf y_0)$ in \eqref{eq:le_app3}. 
    Via observations on function $h_3(\mathbf x_0,\mathbf y_0)$, we can find out that \vspace{-0.1cm}
    \begin{align}
        h_3(\mathbf x_0,\mathbf y_0)=&(h_4(\mathbf x_0,\mathbf y_0)-h_5(\mathbf x_0,\mathbf y_0))^2\nonumber \\
        =&2(h_4(\mathbf x_0,\mathbf y_0))^2+2(h_5(\mathbf x_0,\mathbf y_0))^2-(h_4(\mathbf x_0,\mathbf y_0)+h_5(\mathbf x_0,\mathbf y_0))^2.\label{eq:le_aux1}
    \end{align}
    From the convexity of nonnegative functions $h_4(\mathbf x_0,\mathbf y_0)$ and $h_5(\mathbf x_0,\mathbf y_0)$, we have both $(h_4(\mathbf x_0,\mathbf y_0))^2$ and $(h_5(\mathbf x_0,\mathbf y_0))^2$ are convex. By applying in  \eqref{eq:le_aux1} the first-order condition of convex functions on $(h_4(\mathbf x_0,\mathbf y_0))^2$ and $(h_5(\mathbf x_0,\mathbf y_0))^2$ at the local point $(\mathbf x_0^{(r)},\mathbf y_0^{(r)})$, the left inequality in \eqref{eq:le_app3} can be directly verified. As both 
    $4h_4(\mathbf x_0^{(r)},\mathbf y_0^{(r)})\tilde{h}_4^{(r)}(\mathbf x_0,\mathbf y_0)$ and $4h_5(\mathbf x_0^{(r)},\mathbf y_0^{(r)})\tilde{h}_5^{(r)}(\mathbf x_0,\mathbf y_0)$ in~\eqref{eq:le_app3} are linear functions, the left-hand side function in~\eqref{eq:le_app3} can be proved to be concave. In addition, the first-order expansions at the local point $(\mathbf x_0^{(r)},\mathbf y_0^{(r)})$ also ensure the equality in left-hand side inequality of~\eqref{eq:le_app3} holds at the local point. 

    As for the proof of right-hand side inequality in \eqref{eq:le_app3}, we first have\vspace{-0.1cm}
    \begin{align}
        &\sqrt{h_3(\mathbf x_0,\mathbf y_0)}=|h_4(\mathbf x_0,\mathbf y_0)-h_5(\mathbf x_0,\mathbf y_0)|\nonumber \\
        \!\!\!\!=&\max\{h_4(\mathbf x_0,\!\mathbf y_0)\!-\!h_5(\mathbf x_0,\mathbf y_0),h_4(\mathbf x_0,\mathbf y_0)\!-\!h_5(\mathbf x_0,\mathbf y_0)\!\}.\!\!\label{eq:le_aux2}
    \end{align}
    Further considering the convexity of both $h_4(\mathbf x_0,\mathbf y_0)$ and $h_5(\mathbf x_0,\mathbf y_0)$, we have \vspace{-0.1cm}
    \begin{align}
        h_4(\mathbf x_0,\mathbf y_0)-h_5(\mathbf x_0,\mathbf y_0)\leq h_4(\mathbf x_0,\mathbf y_0)-\tilde{h}_5^{(r)}\!(\mathbf x_0,\mathbf y_0),\label{eq:le_aux3}\\
        h_5(\mathbf x_0,\mathbf y_0)-h_4(\mathbf x_0,\mathbf y_0)\leq h_5(\mathbf x_0,\mathbf y_0)-\tilde{h}_4^{(r)}\!(\mathbf x_0,\mathbf y_0).\label{eq:le_aux4}
    \end{align}
    Inserting \eqref{eq:le_aux3} and \eqref{eq:le_aux4} into \eqref{eq:le_aux2}, we have \vspace{-0.1cm}
    \begin{align}
        &\sqrt{h_3(\mathbf x_0,\mathbf y_0)}
        \leq\max\{h_4(\mathbf x_0,\mathbf y_0)\!-\!\tilde{h}_5^{(r)}\!(\mathbf x_0,\mathbf y_0),h_5(\mathbf x_0,\mathbf y_0)\!-\!\tilde{h}_4^{(r)}\!(\mathbf x_0,\mathbf y_0)\!\},\label{eq:le_aux5}
    \end{align}

    \vspace{-0.2cm}
    \noindent{}which also ensures the right-hand side term in \eqref{eq:le_aux5} to be nonnegative and convex, as well as the convexity of right-hand side function in \eqref{eq:le_app3}. Further applying square operation on both sides of \eqref{eq:le_aux5} will lead to the right-hand side inequality in \eqref{eq:le_app3}. Similarly, the first-order Taylor expansions have guaranteed the equality holds at the local point $(\mathbf x_0^{(r)},\mathbf y_0^{(r)})$. 
\end{proof}
\vspace{-0.2cm}
Clearly, the function ${\psi}_{\mathbf s,n_1,n_2}(\mathbf x,\mathbf y)$ can be transferred into the same format as $h_3(x_1,y_1,x_2,y_2)$ with $x_1=x_{n_1}-s_x$, $y_1=y_{n_1}-s_y$, $x_2=x_{n_2}-s_x$ and $y_2=y_{n_2}-s_y$. 
Thus, via applying the inequality \eqref{eq:le_app3} in Lemma~\ref{le:crb_app2}, we can construct two functions, i.e., $\tilde{\psi}_{\mathbf s,n_1,n_2}^{(r)}(\mathbf x,\mathbf y)$ and $\hat{\psi}_{\mathbf s,n_1,n_2}^{(r)}(\mathbf x,\mathbf y)$, such that\vspace{-0.1cm}
\begin{equation}\vspace{-0.1cm}
    \tilde{\psi}_{\mathbf s,n_1,n_2}^{(r)}(\mathbf x,\mathbf y)\leq {\psi}_{\mathbf s,n_1,n_2}(\mathbf x,\mathbf y)\leq \hat{\psi}_{\mathbf s,n_1,n_2}^{(r)}(\mathbf x,\mathbf y),
\end{equation}
where $\tilde{\psi}_{\mathbf s,n_1,n_2}^{(r)}(\mathbf x,\mathbf y)$ is concave, $\hat{\psi}_{\mathbf s,n_1,n_2}^{(r)}(\mathbf x,\mathbf y)$ is convex, and both equalities hold at local point $(\mathbf x^{(r)},\mathbf y^{(r)},\mathbf a^{(r)})$. By replacing the two occurrences of ${\psi}_{\mathbf s,n_1,n_2}(\mathbf x,\mathbf y)$ in \eqref{eq:crb_c_app2} respectively  with $\hat{\psi}_{\mathbf s,n_1,n_2}^{(r)}(\mathbf x,\mathbf y)$ and $\tilde{\psi}_{\mathbf s,n_1,n_2}^{(r)}(\mathbf x,\mathbf y)$, we can accomplish the concave approximation as \vspace{-0.1cm}
\begin{align}
	&\big(\frac{\eta^2}{d_{n_1}^6d_{n_2}^6}+\frac{8\eta}{d_{n_1}^6d_{n_2}^4}+\frac{8\eta}{d_{n_1}^4d_{n_2}^6}+\frac{64}{d_{n_1}^4d_{n_2}^4}\big)\cdot \psi_{\mathbf s,n_1,n_2}(\mathbf x,\mathbf y)\nonumber \\
	\geq &-\frac{{E}_{4,n_1,n_2}^{(r)}}{2}\big({E}_{1,n_1,n_2}^{(r)}d_{n_1}^2+{E}_{2,n_1,n_2}^{(r)}d_{n_2}^2\big)^2-\frac{1}{2{E}_{4,n_1,n_2}^{(r)}\!\!}\hat{\psi}_{\mathbf s,n_1,n_2}(\mathbf x,\mathbf y)^2\!+\!E_{3,n_1,n_2}^{(r)}\tilde{\psi}_{\mathbf s,n_1,n_2}(\mathbf x,\mathbf y)\nonumber \\
    \triangleq&g_{c,n_1,n_2,\mathbf s}^{(r)}(\mathbf x,\mathbf y), \label{eq:crb_c_app3}
\end{align}

\vspace{-0.2cm}
    \noindent{}where the positive constants $E_{3,n_1,n_2}^{(r)}$ and $E_{4,n_1,n_2}^{(r)}$ have ensured the inequality in \eqref{eq:crb_c_app3} still holds. The convexity of non-negative function $\hat{\psi}_{\mathbf s,n_1,n_2}^{(r)}(\mathbf x,\mathbf y)$ and the concavity of function $\tilde{\psi}_{\mathbf s,n_1,n_2}^{(r)}(\mathbf x,\mathbf y)$ result in the concavity of the eventually  obtained function $g_{c,n_1,n_2,\mathbf s}^{(r)}(\mathbf x,\mathbf y)$ in \eqref{eq:crb_c_app3}. 
Further according to Lemma~\ref{le:crb_app2}, we have the equality in \eqref{eq:crb_c_app3} also holds at local point $(\mathbf x^{(r)},\mathbf y^{(r)},\mathbf a^{(r)})$. 

As a result, we can accomplish the concave approximation construction for $\Theta_{a,\mathbf s}^{[m]}\Theta_{b,\mathbf s}^{[m]}-(\Theta_{c,\mathbf s}^{[m]})^2$ as\vspace{-0.1cm}
\begin{align}
	\Theta_{a,\mathbf s}^{[m]}\Theta_{b,\mathbf s}^{[m]}-(\Theta_{c,\mathbf s}^{[m]})^2&
	\geq \sum_{n_1=m}^{m+L-2}\sum_{n_2=n_1+1}^{m+L-1}g_{c,n_1,n_2,\mathbf s}^{(r)}(\mathbf x,\mathbf y)\nonumber\\
    &\triangleq g_{2,m,\mathbf s}^{(r)}(\mathbf x,\mathbf y),
\end{align}

\vspace{-0.2cm}
    \noindent{}where the equality is always guaranteed at local point  $(\mathbf x^{(r)},\mathbf y^{(r)},\mathbf a^{(r)})$. 
\vspace{-0.1cm}
\subsection{Iterative Algorithm}
Via replacing the objective in (P1) and the constraint \eqref{eq:p1_con3} respectively with the corresponding concave and convex approximation, we can obtain the following subproblem established based on local point $(\mathbf x^{(r)},\mathbf y^{(r)},\mathbf a^{(r)})$:\vspace{-0.1cm}
\begin{subequations}
	\begin{align}
		(\text{P2}): \max_{\mathbf x,\mathbf y,\mathbf a} &~~ \min_{k\in\mathcal K}\{\sum_{n=1}^Nf_{n,k}^{(r)}(\mathbf x,\mathbf y,\mathbf a)\} \\
		\mathrm{s.t.}~&~~(x_{n+1}\!-\!x_{n})^2\!+\!(y_{n+1}\!-\!y_{n})^2\!\leq\! \delta^2V^2,\forall n,\\
		&~~(x_n,y_x)\in\Omega,~\forall n\in\mathcal N,\\
		&~~g_{1,m,\mathbf s}^{(r)}(\mathbf x,\mathbf y)-\xi_lg_{2,m,\mathbf s}^{(r)}(\mathbf x,\mathbf y)\leq 0,~\forall m\in\mathcal N,~\forall \mathbf s\in\mathcal F,\\
		&~~ \sum_{k=1}^Ka_{k,n}=1,~\forall n\in\mathcal N,\\
		&~~ 0\leq a_{k,n}\leq 1,~\forall k\in\mathcal K,~\forall n\in\mathcal N,
	\end{align}
\end{subequations}

\vspace{-0.2cm}
    \noindent{}which is a convex problem. 
Note that the approximation and inequality we have guaranteed in the entire approximation process, have ensured that optimizing (P2) will lead to an objective improvement with respect to the local point $(\mathbf x^{(r)},\mathbf y^{(r)},\mathbf a^{(r)})$ without violating the constraints. This implies that updating the local point via iteratively solving (P2) will lead to a converged superior solution. 

Hence, based on the above discussions, we are motivated to propose an iterative algorithm for efficiently addressing the problem (P1). More specifically, we start with an initial local point  $(\mathbf x^{(0)},\mathbf y^{(0)},\mathbf a^{(0)})$. For each iteration $r$, we construct the convex approximation (P2) based on local point  $(\mathbf x^{(r)},\mathbf y^{(r)},\mathbf a^{(r)})$. After solving (P2), the obtained better solution will be applied as the local point  $(\mathbf x^{(r+1)},\mathbf y^{(r+1)},\mathbf a^{(r+1)})$ for the next iteration. In such manner, the objective value at the local point will be constantly improved without violating any on-demand sensing constraints. Since the objective is upper bounded, the algorithm will eventually converge to a suboptimal solution. 

Hereby, we also provide a complexity analysis on our proposed algorithm. 
Following the complexity analysis in \cite{free_space1}, we have in each iteration the computational complexity for solving convex approximation (P2) with $(K+2)N$ variables is given by $\mathcal O((K+2)^4N^4)=\mathcal O (K^4N^4)$. Denoting by $\varepsilon$ the iteration number needed for convergence, we can obtain the complexity of our proposed iterative algorithm as $\mathcal O (\varepsilon K^4N^4)$.
\vspace{-0.1cm}
\subsection{Feasibility Discussion and  Trajectory Initialization}
\vspace{-0.1cm}
Notice that our proposed iterative algorithm relies on a feasible initial point  $(\mathbf x^{(0)},\mathbf y^{(0)},\mathbf a^{(0)})$ for (P1). However, in practice, the problem (P1) may not be feasible, on account of the strict on-demand sensing constraints \eqref{eq:p1_con2} and \eqref{eq:p1_con3}. Even when (P1) is feasible, it may still be strenuous to obtain a feasible initial point for (P1). In this subsection, to cope with these issues, we present an efficient approach for the feasibility inspection of problem (P1), where a feasible initial solution will be constructed in case with feasible (P1). 

Via observing (P1), we can find out that a solution $(\mathbf x,\mathbf y,\mathbf a)$ fulfilling all constraints of (P1) except \eqref{eq:p1_con3} can be easily constructed. Thus, the key point in feasibility discussion is to determine if we can construct a feasible UAV trajectory fulfilling the constraint \eqref{eq:p1_con3}. 
As a result, the feasibility of problem (P1) is equivalent to the case having the following problem (P3) with a negative objective.\vspace{-0.1cm}
\begin{align}
		(\text{P3}): \min_{\mathbf x,\mathbf y} &~ \max_{m\in\mathcal N,\mathbf s\in\mathcal F}\{\Theta_{a,\mathbf s}^{[m]}\!+\!\Theta_{b,\mathbf s}^{[m]}\!-\! \xi_l(\Theta_{a,\mathbf s}^{[m]}\Theta_{b,\mathbf s}^{[m]}\!-\!(\Theta_{c,\mathbf s}^{[m]})^2)\} \\
		\mathrm{s.t.}&~\eqref{eq:p1_con1}~\text{and}~\eqref{eq:p1_con2}.\nonumber 
\end{align}

\vspace{-0.2cm}
\noindent{}Since (P3) is clearly a nonconvex problem, we can directly adopt the convex approximations constructed in previous subsections and approach to a feasible initial solution of (P1) in a similarly iterative manner. Specifically, we start with a UAV trajectory  feasible for (P3). In each iteration, we construct and solve the following convex approximation for (P3).\vspace{-0.1cm}
\begin{align}
		(\text{P4}): \min_{\mathbf x,\mathbf y} &~ \max_{m\in\mathcal N,\mathbf s\in\mathcal F}\{g_{1,m,\mathbf s}^{(r)}(\mathbf x,\mathbf y)-\xi_lg_{2,m,\mathbf s}^{(r)}(\mathbf x,\mathbf y)\} \\
		\mathrm{s.t.}&~\eqref{eq:p1_con1}~\text{and}~\eqref{eq:p1_con2}.\nonumber 
\end{align}
If iteratively solving (P4) converges to a solution with positive objective in (P3), we can then identify the problem (P1) as infeasible. Conversely, in each iteration, once we get a UAV trajectory with negative objective in (P4), the guaranteed inequalities in our proposed process for approximation constructions will ensure the trajectory also resulting in negative objective in (P3), such that a feasible UAV trajectory $(\mathbf x^{(0)},\mathbf y^{(0)})$ for (P1) is obtained. Initializing user assignment $\mathbf a^{(0)}$ with all entries equal to $\frac{1}{K}$, we can then obtain a feasible initial point $(\mathbf x^{(0)},\mathbf y^{(0)},\mathbf a^{(0)})$ for (P1).

The overall algorithm flow for addressing (P1) has been summarized in Algorithm~\ref{algorithm1}.

\begin{algorithm}[!t]\small
	\algsetup{linenosize=\small}
	\caption{\bf{ Overall Algorithm Flow for Addressing (P1)}}
	\begin{algorithmic}
		\STATE \noindent{\bf{$\!\!\!\!\!\!$Feasibility Inspection and Initialization}} \\
		\STATE   \noindent{\bf{a)}} ~~Initialize a feasible UAV trajectory $(\hat{\mathbf x},\hat{\mathbf y})$ for (P3). 
		\STATE \noindent{\bf{b)}} ~~Construct (P4) based on local point $(\hat{\mathbf x},\hat{\mathbf y})$.
        \STATE \noindent{\bf{c)}} ~~Solve (P4) and obtain new point $(\hat{\mathbf x}^*,\hat{\mathbf y}^*)$.
        \STATE \noindent{\bf{d)}} ~~{\bf If}~~(P3) objective at $(\hat{\mathbf x}^*,\hat{\mathbf y}^*)$ is negative 
        \STATE ~~~~~~~~~ Define $(\mathbf x^{(0)},\mathbf y^{(0)})=(\hat{\mathbf x}^*,\hat{\mathbf y}^*)$ and go to {\bf e)}.
        \STATE ~~~~ \noindent{\bf{Elseif}}~~(P3) objective reduction is sufficiently small \\
		
	\STATE ~~~~~~~~~ (P1) is infeasible and stop the algorithm.\\ 
		\STATE ~~~~ \noindent{\bf{Else}} \\
		
		\STATE ~~~~~~~~~ Update $(\hat{\mathbf x},\hat{\mathbf y})=(\hat{\mathbf x}^*,\hat{\mathbf y}^*)$ and back to {\bf b)}.\\
		
		
		\STATE ~~~~ \noindent{\bf{End}}
        \STATE   \noindent{\bf{e)}} ~~Initialize $\mathbf a^{(0)}$ and set $r=0$.
		\STATE \noindent{ \bf{$\!\!\!\!\!\!\!\!$Iteration}} \\
		
		\STATE \noindent{\bf{f)}} ~~Construct 
		(P2) 		based on $(\mathbf x^{(r)},\mathbf y^{(r)},\mathbf a^{(r)})$.
		
		\STATE \noindent{\bf{g)}} ~~Solve (P2) and obtain $(\mathbf x^{(r*)},\mathbf y^{(r*)},\mathbf a^{(r*)})$.\\ 
		
		\STATE \noindent{\bf{h)}} ~ \noindent{\bf{If}} ~   (P1) objective improvement is sufficiently small\\
		\STATE ~~~~~~~~~ Define final solution $(\mathbf x^{*},\mathbf y^{*},\mathbf a^{*})\!=\!(\mathbf x^{(r*)},\mathbf y^{(r*)},\mathbf a^{(r*)})$. 
		\STATE ~~~~~~~~~ Stop the algorithm.
		\STATE ~~~~ \noindent{\bf{Else}} \\
		
		\STATE ~~~~~~~~~ Update $(\mathbf x^{(r+1)},\mathbf y^{(r+1)},\mathbf a^{(r+1)})=(\mathbf x^{(r*)},\mathbf y^{(r*)},\mathbf a^{(r*)})$.
		
		
		\STATE ~~~~~~~~~ $r=r+1$ and back to \bf{f)}.
		
		\STATE ~~~~ \noindent{\bf{End}} \\
	\end{algorithmic}
	\label{algorithm1}
\end{algorithm}
\vspace{-0.2cm}
\section{Simulation Results}
This section provides numerical results to evaluate the performance of our proposed optimization scheme for joint trajectory and user assignment design scheme in the considered UAV-enabled ISAC networks. 
Unless otherwise statements, the network in the simulations is build up in a 400m $\times$ 400m square area which distributes $K=5$ communication users and a on-demand sensing region located at the center of the area with radius $r=50$m. Other parameters are set as: $T=100$s, $V=10$m/s, $H=20$m, $N=25$, $L=5$, $\beta=-60$dB, $P=20$dBm, $\sigma^2=\sigma_0^2=-100$dBm, $\sigma_rG_r=53$dB, 
$a_t=100$, $\xi_l=10$m$^2$, $\bar d_0=250$m.
\vspace{-0.3cm}
\subsection{Characterized Results of On-Demand Sensing}
\vspace{-0.1cm}
As discussed in Section III, we have characterized the on-demand detection constraint as a restricted deployment area for UAV and the on-demand localization constraint as localization CRB constraints over finite reference target points. In this subsection, we first depict the characterization results under different sensing accuracy requirements.  

In Fig.~\ref{detection_st}, over a nonconvex sensing region, we demonstrate 
the characterized UAV deployment region 
with different transmit power levels $P$. As depicted in Fig.~\ref{detection_st}, the simulated nonconvex sensing region consists of two non-intersecting circular areas. 
The corresponding deployment region ${\Omega}$ for UAV shows to be convex in 
Fig.~\ref{detection_st}, 
\begin{figure}[!t]
	\centering
	\includegraphics[width=0.6\linewidth, trim= 0 20 0 20 ]{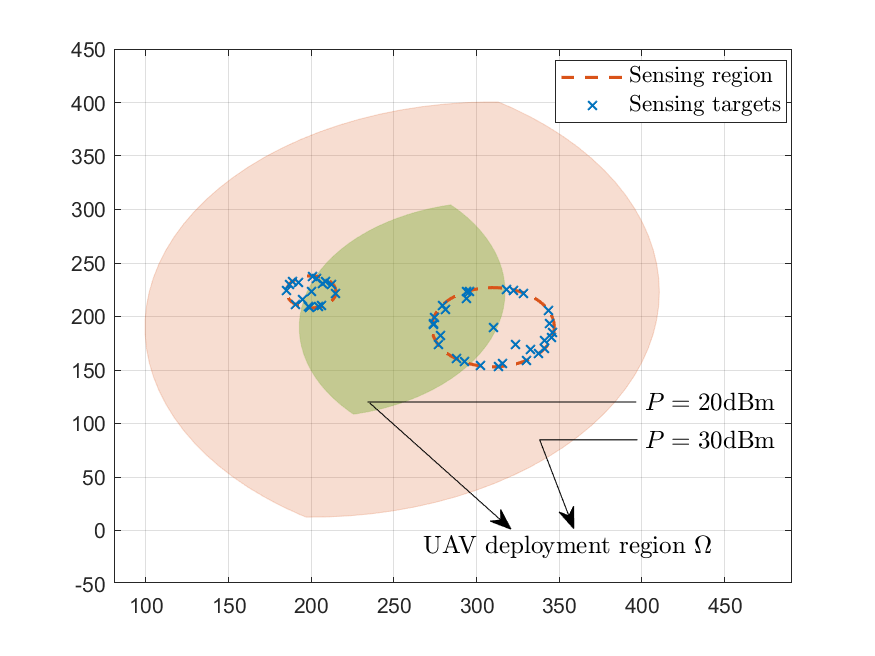}
	\caption{UAV deployment restriction owing to on-demand detection constraints.}\label{detection_st}
		\vspace{-.1cm}
\end{figure}
confirming the statements in our proposed Lemma~\ref{le:detect_region}. Furthermore, the green region corresponds to the UAV deployment region under a lower transmission power, $P=20$dBm, where a lower SNR results in a smaller UAV detection area, thus limiting the degree of freedom for UAV movement. In contrast, the pink region represents the deployment region ${\Omega}$ when $P=30$dBm. With higher transmit power, the accordingly expanded UAV
deployment area 
allows the UAV to execute tasks more freely. In addition, the reference sensing targets for on-demand localization have been depicted in Fig.~\ref{detection_st}, distributed on both circular serving areas. 

Next, we more deeply explore the influence of the maximum allowable localization CRB $\xi_l$ on the required number of reference sensing targets. 
\begin{figure}[!t]
	\centering
	\includegraphics[width=0.6\linewidth, trim= 0 20 0 20 ]{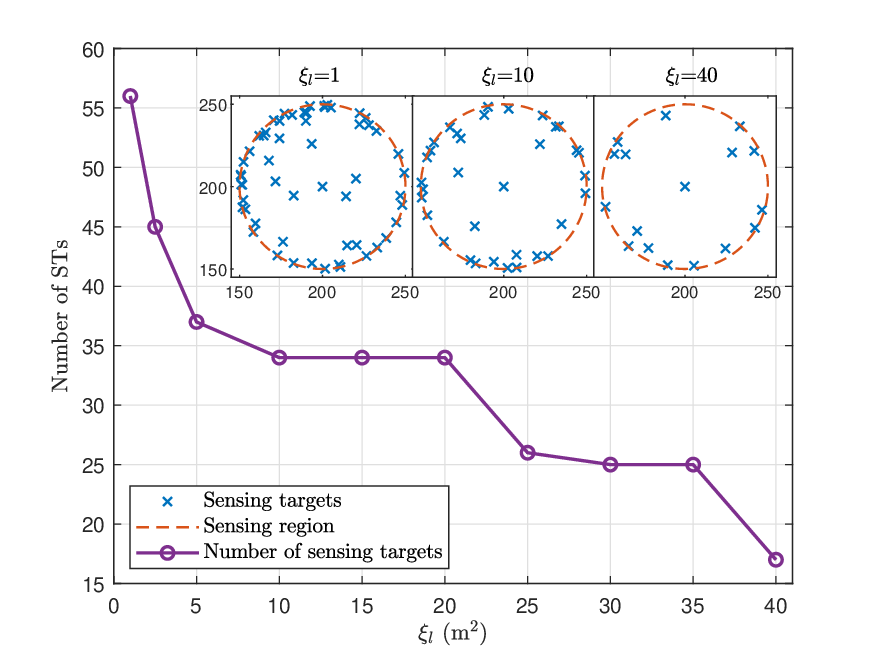}
	\caption{Requested number of reference sensing targets for on-demand localization with different $\xi_l$.}\label{crb_st}
		\vspace{-.5cm}
\end{figure}
As indicated in Fig.~\ref{crb_st}, when the CRB requirement $\xi_l$ decreases, the number of reference sensing targets increases, which reflects that a stricter CRB constraint necessitates more reference sensing targets to guarantee the localization QoS. On the other hand, when the CRB constraint is more relaxed, fewer sensing targets are needed, as a certain-level degradation of localization accuracy can be tolerated. The three inset images provide clearer visual comparisons of requested reference sensing targets at different levels of $\xi_l$, with the target density decreasing in $\xi_l$. 
In particular, we can notice that the reference sensing targets are more allocated near the region boundary, showing that boundary points are more like the bottleneck of localization accuracy. Even though, together with these boundary points, a few interior reference points are still needed to fully represent the whole region in on-demand localization accuracy evaluations.
\vspace{-0.2cm}
\subsection{Validation of Proposed Optimization Scheme}
\vspace{-0.1cm}
Afterwards, to reveal the significant benefits of our proposed optimization 
scheme for UAV-enabled ISAC with on-demand sensing, we introduce 
the following scheme proposed in \cite{Jing_1,Gu_1} as the benchmark. 
\begin{itemize}
    \item {\bf Adj gradient descend:}     
    To deal with localization CRB $\Phi_{\mathbf s}^{[m]}(\mathbf x,\mathbf y)$, at a local point 
    in each iteration, this scheme approximates the complicated function $\Phi_{\mathbf s}^{[m]}$ as 
    the first-order Taylor expansion~\cite{Jing_1,Gu_1}, 
    thereby achieving a convex approximation of the non-convex localization CRB constraints. Ignoring the approximation errors in higher-order Taylor expansion terms, convex optimization tools are employed to solve the approximated problem. 
    Then, the direction from the obtained solution to the local point 
    is taken as the descent gradient. 
    An adjustable step size 
    is selected along the  gradient to maximize the objective function, fulfilling the solution feasibility to a certain extent. The obtained solution will be reapplied as the local point for the next iteration until convergence. However, when the local point reaches the boundary of feasible set, the descend direction is very likely leading all adjusted solutions infeasible, such that the the algorithm will be interrupted at a low-performance point.  
    \item {\bf Fix gradient descend:} Note that the Adj gradient descend scheme 
    may 
    choose an excessively large step size in the early stages of iterations, leading to convergence to a low-performance point or an infeasible solution. 
    We also provide the scheme with always fixed step size as the second benchmark, which can potentially achieve a better solution than Adj gradient descend with much more iteration rounds. Even though, the iterations in both existing benchmarks may be interrupted due to reaching the infeasible region. 
\end{itemize}

We first illustrate the convergence behavior of our proposed scheme and two benchmarks in Fig.~\ref{convergnce}. 
\begin{figure}[!t]
 \centering
    \subfigure[]{
  \includegraphics[width=0.6\linewidth, trim=0 10 0 10]{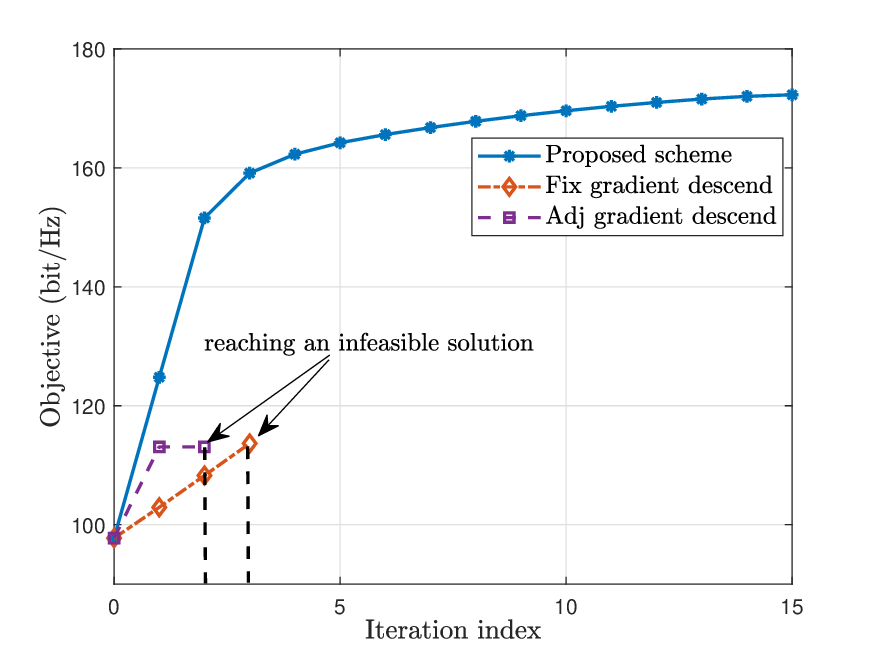}}\vspace{0.05cm}
 \subfigure[]{
  \includegraphics[width=0.6\linewidth, trim=0 10 0 10]{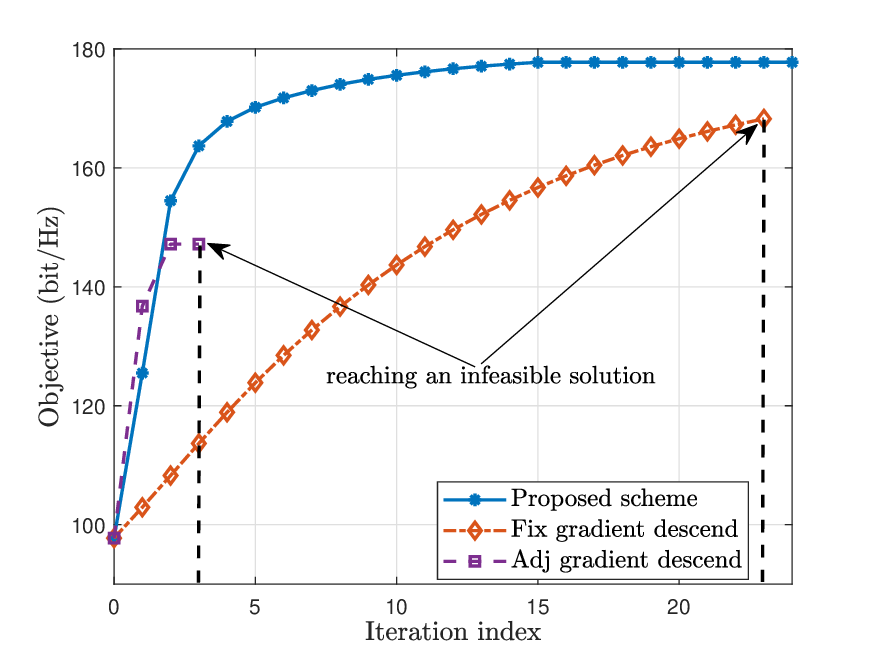}}
   \vspace{-.2cm}
 \caption{Convergence behavior of  proposed scheme and benchmarks with  $\xi_l=2.5$ in (a) and $\xi_l=5$ in (b).}
\label{convergnce}
 \vspace{-.5cm}
\end{figure}
In particular, we provide the objective value after each iteration in Fig.~\ref{convergnce}(a) for $\xi_l=2.5$ and in Fig.~\ref{convergnce}(b) for $\xi_l=5$. It can be observed that the communication performances achieved by all three schemes experience an increase in iterations, while the proposed scheme converges to a much better solution outperforming benchmarks. 
For both benchmarks simply applying the first-order Taylor approximation to deal with localization CRB, it is evident that Adj gradient descend scheme surpasses Fix gradient descend scheme in terms of convergence speed, while it ultimately under-performs Fix gradient descend in final converged performance. This is due to the fact that  although an adjustable step size may enable the algorithm to select a local optimal solution with better performance during the initial iterations, the feasible region error introduced by the non-tight upper bound convex approximation can lead to a higher likelihood of becoming trapped in a local optimum or encountering infeasible solutions in subsequent iterations. Furthermore, Unlike the proposed scheme which exhibits favorable steady convergence behavior, we can see that the iterations of benchmarks are all stopped after only certain counts, ending up at a position with very poor performance where such phenomenon is more pronounced under stricter CRB constraints $\xi_l=2.5$. 

The phenomenon of interrupted iterations in both benchmarks, shown in Fig.~\ref{convergnce}, is due to non-tight upper bound convex approximation which extends the feasible region impractically and may lead to the premature termination of iterations due to the infeasibility of solutions. To verify this issue and reveal the benefits of our proposed scheme in guaranteeing on-demand localization CRB, we provide  in Fig.~\ref{optimized crb} the achieved maximum localization CRBs for 
all three schemes over different time slots, respectively for cases $\xi_l=2.5$ and $\xi_l=5$. 
\begin{figure}[!t]
 \centering
    \subfigure[]{
  \includegraphics[width=0.6\linewidth, trim=0 10 0 0]{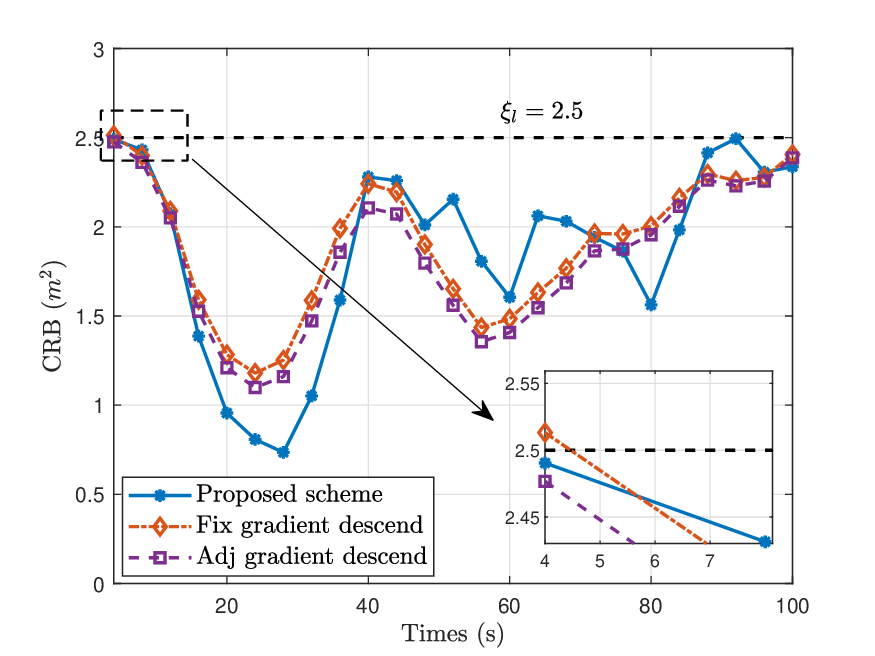}}
    \subfigure[]{
  \includegraphics[width=0.6\linewidth, trim=0 10 0 0]{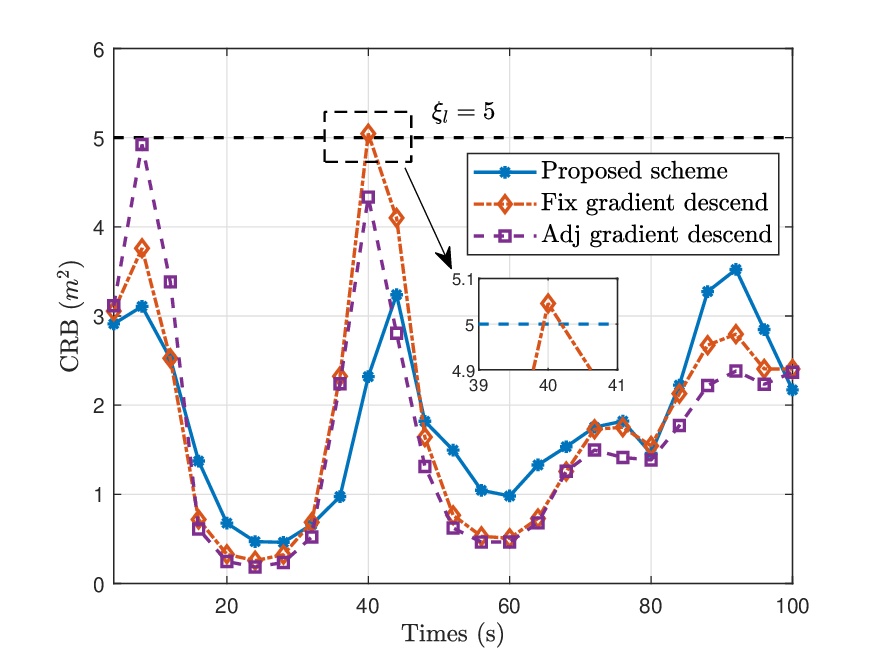}}
 \caption{Localization CRBs of optimized solution of proposed scheme and benchmarks over time when $\xi_l=2.5$ and $\xi_l=5$.}
\label{optimized crb}
 \vspace{-.1cm}
\end{figure}
In Fig.~\ref{optimized crb}, the CRBs of our proposed scheme always satisfies the CRB limitation $\xi_l$, as the proposed tightly upper-bound convex approximation ensures solution feasibility 
in each iteration. By contrast, the CRBs of Fix gradient descend scheme are over the CRB limitation at certain time which is precisely the reason for the premature termination of iteration. It is worthwhile to note that although the CRBs of Adj gradient descend scheme are below $\xi_l$, all points along the descent gradient from optimized solution to local point are all infeasible except the local point, such that the last iteration of Adj gradient descend scheme stops at the low-performance point, as shown in Fig.~\ref{convergnce}. With relaxed CRB limitation exhibiting greater tolerance for such imprecise approximations, the benchmarks can undergo more iterations before reaching an infeasible solution, but still with worse performance than our proposed scheme. The above observations substantiate the superiority of our proposed scheme in terms of steady convergence and localization CRB guarantee.

\subsection{Performance Evaluation of Proposed Solution}
\vspace{-0.1cm}
Finally, we move on to evaluating the optimized communication performance of our proposed scheme under different setups, also in comparison with both benchmarks. 
To start with, Fig.~\ref{crb_value} evaluates the impact of CRB limitations $\xi_l$ on optimized communication performance. 
\begin{figure}[!t]
	\centering
	\includegraphics[width=0.6\linewidth, trim= 0 20 0 20 ]{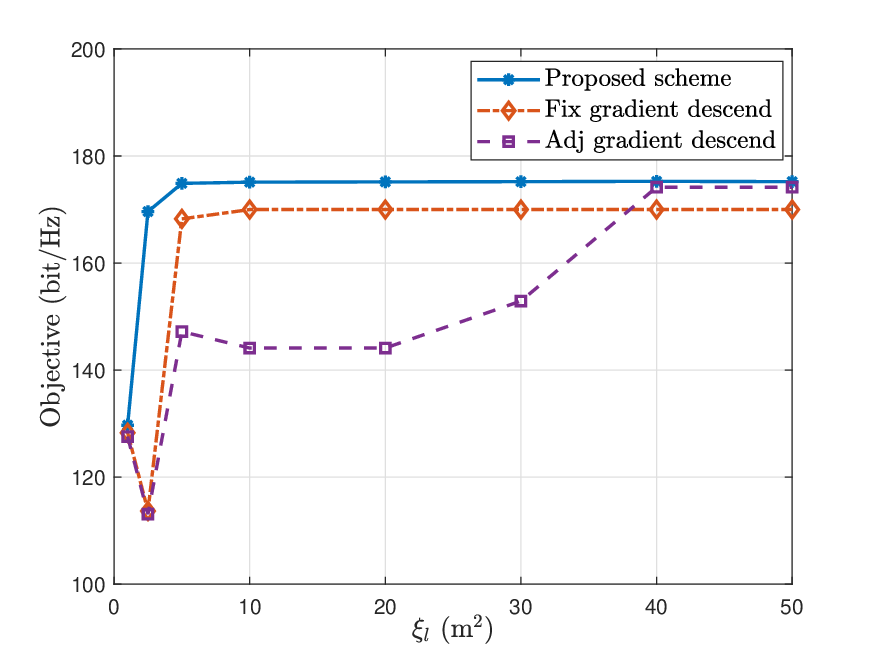}
	\caption{Communication performance comparison under different CRB limitations $\xi_l$.}\label{crb_value}
		\vspace{-.5cm}
\end{figure}
As illustrated, 
when the value of $\xi_l$ increases, i.e., the CRB constraint is looser, the communication performance of our proposed scheme increases smoothly and always outperforms 
that of benchmarks. 
When the $\xi_l$ is relatively small, relaxing CRB constraint significantly enhances the degree of freedom for UAV movement, thereby markedly improving communication performance. When $\xi_l$ is getting sufficiently large, the critical factor restricting the degree of freedom for UAV movement becomes the target detection constraint \eqref{eq:op_con2_new}. At this point, increasing the value of $\xi_l$ does not substantially affect the optimized trajectory of the UAV, as well as the communication performance. However, unlike the proposed scheme showing stable performance improvement with respect to 
$\xi_l$, the communication performances of the two benchmarks have shown significant fluctuations, depicted in Fig.~\ref{crb_value}. 
This is also due to the improper approximations for localization CRB, 
which leads benchmarks highly susceptible to get trapped into low-performance local points. 
This phenomenon is more pronounced when the value of $\xi_l$ is smaller, indicating higher requirements for positioning performance, which suggests that the benchmarks have application limitations in scenarios demanding high positioning accuracy. 

To reveal more deep insights on the performance differences in Fig.~\ref{crb_value}, we present in Fig.~\ref{trajectory} 
the optimized UAV trajectories of our proposed scheme and the two benchmarks, under different CRB requirements $\xi_l$. 
\begin{figure}[!t]
 \centering
    \subfigure[$\xi_l=2.5$]{
 \vspace{-.2cm}
  \includegraphics[width=0.6\linewidth, trim=0 10 0 0]{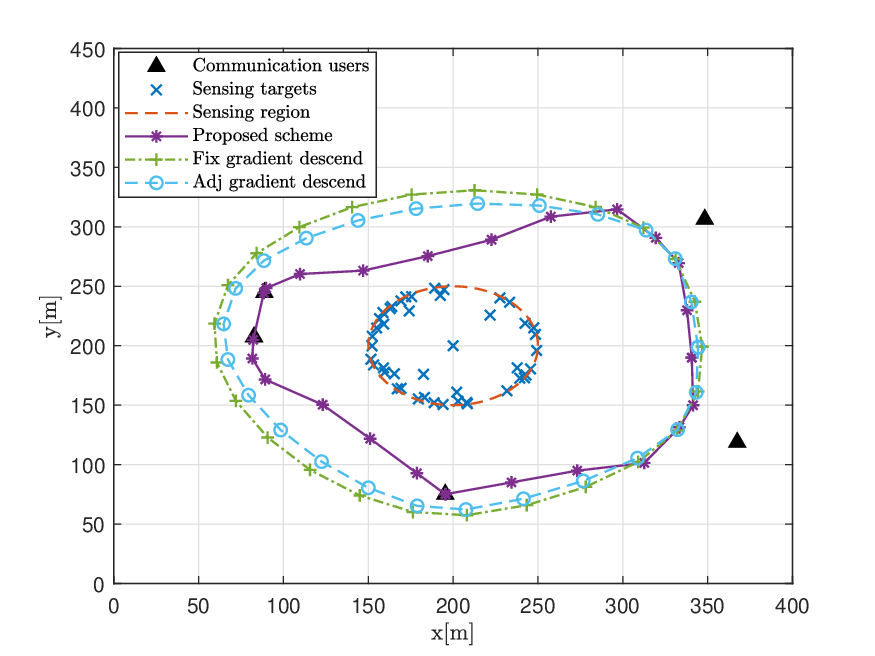}}
 \subfigure[$\xi_l=10$]{
  \includegraphics[width=0.6\linewidth, trim=0 10 0 0]{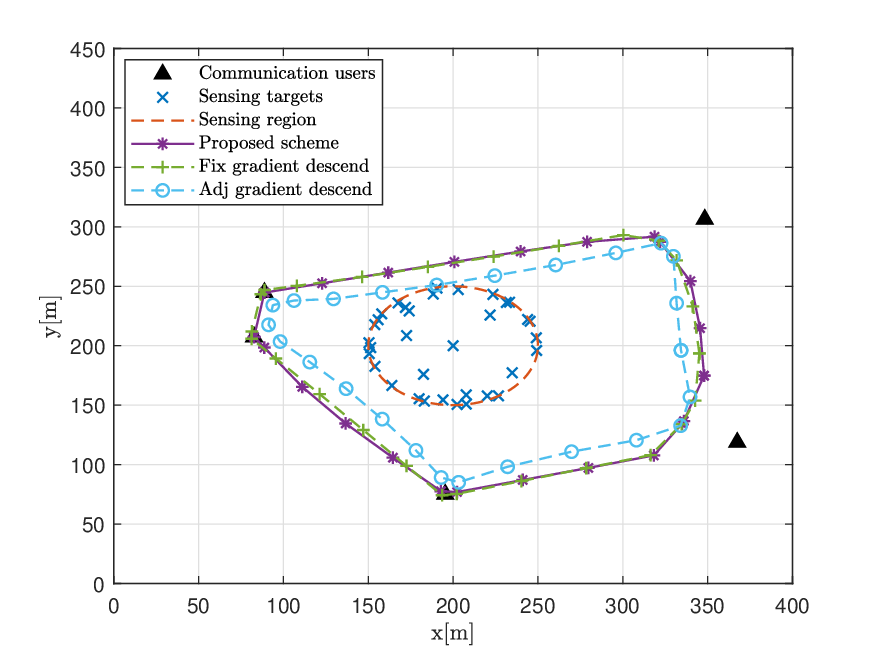}}
   \vspace{-.2cm}
 \caption{Optimized UAV trajectories of different schemes with $\xi_l=2.5$ in (a) and $\xi_l=10$ in (b).}
\label{trajectory}
 \vspace{-.5cm}
\end{figure}
In both Fig.~\ref{trajectory}(a) and Fig.~\ref{trajectory}(b), the trajectories from our proposed scheme have been optimized closer to the communication users while strictly fulfilling the on-demand sensing constraints, such that obvious performance advantages of our proposed scheme can be observed in Fig.~\ref{crb_value}. Conversely, due to the improper CRB approximation, the two benchmarks in Fig.~\ref{trajectory}(a), easily get trapped into local trajectory solutions, showing much worse communication performance owing to the longer communication distances. Similar conclusion can be found out in Fig.~\ref{trajectory}(b) for the trajectory optimized via Adj gradient descend scheme. Particularly observing the optimized trajectory from our proposed scheme, we can discover that while CRB requirement increases from $\xi_l=2.5$ to $\xi_l=10$, the UAV trajectory exhibits notable changes reflecting the trade-off between communication performance and sensing accuracy. In particular, 
the trajectory of the proposed scheme becomes more attracted to the sensing region in Fig.~\ref{trajectory}(a) with a stricter CRB requirement, and more tightened towards communication users in Fig.~\ref{trajectory}(b) with a lager maximum allowable CRB value, indicating an interesting shift of optimization focus from sensing to communication. 
As shown in Fig.~\ref{trajectory}(a), corresponding to a stricter CRB requirement, $\xi_l=2.5$, the proposed scheme demonstrates a 
UAV flight path adjusted 
towards the sensing region, striking a balance between sensing and communication needs, to adapt to the stricter CRB constraint. 
When the CRB constraint is relaxed to $\xi_l=10$ in Fig.~\ref{trajectory}(b), the requirement for localization accuracy becomes less stringent, allowing more freedom in UAV movement 
with reduced focus on sensing region, 
resulting in the more straight trajectory between communication users. 
The above observations highlight that our proposed scheme is capable of effectively balancing the trade-off between communication quality and sensing accuracy, outperforming the benchmarks in throughput while maintaining satisfactory levels of sensing accuracy.

At last, in Fig.~\ref{targeted_region_value}, we 
demonstrate the communication performance of all three schemes, 
under different range sizes of targeted region $\mathcal{D}$. 
\begin{figure}[!t]
	\centering
	\includegraphics[width=0.6\linewidth, trim= 0 20 0 20 ]{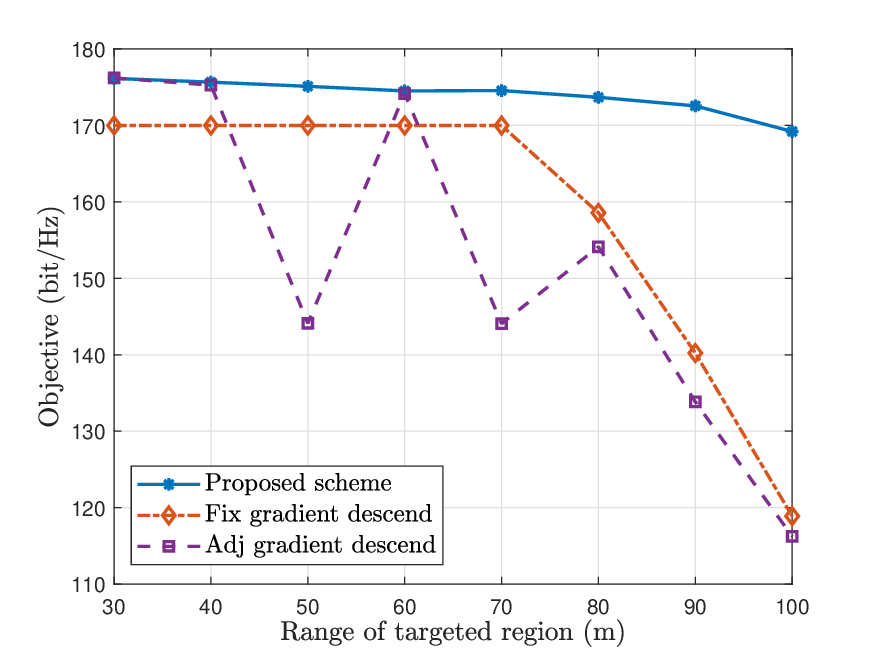}
	\caption{Communication performance comparison under different targeted region size.}\label{targeted_region_value}
		\vspace{-.5cm}
\end{figure}
It can be seen from this figure that with enlarged serving region $\mathcal D$, the communication performance of our proposed scheme exhibits a gradual decline. Indeed, a larger region $\mathcal D$ induces stricter sensing constraints, such that the flight of UAV is more restricted, 
leading to a degradation in the service performance for communication users. 
Compared to our proposed scheme, the benchmarks are much  more sensitive to the variation of targeted region. Especially when the region is sufficiently large, further increasing the sensing region size will result in a sharp decline in their communication performance. Additionally, it can be observed that under certain parameter settings, 
the Adj gradient descend scheme can potentially achieve performance very close to that of our proposed scheme. But, under most parameter settings, its performance is the worst among all three schemes. This is also due to the non-strict constraint approximation and adjustable iteration step sizes, which may cause that the algorithm converges to a high-quality solution in the initial few iterations but are more likely to trap it into local optimal solution, resulting in poor algorithmic stability. The above observations not only reveal the limitation of benchmarks in addressing the system design for large-scale ISAC networks, but also demonstrates the powerful capability of our proposed scheme in balancing the demands of sensing and communication.

\section{Conclusion}
In this work, we have investigated a UAV-enabled ISAC network, 
where the UAV follows a periodical circular trajectory to provide continuous on-demand sensing service support toward ground targets. We considered both on-demand detection and on-demand target localization in UAV-enabled ISAC, and focused on enhancing the communication performance while strictly fulfilling the accuracy constraints of these on-demand sensing tasks. To deal with the resulted infinite sensing constraints, we formulated the on-demand detection constraint to a deployment area constraint for UAV and devised finite reference target points sufficiently representing the whole serving region for on-demand localization tasks. 
Next, to address the reformulated problem with high optimization difficulties, we constructed a tight convex approximation for the whole problem, which guaranteed strict fulfillment of these sensing constraints. Based on that, we proposed an iterative algorithm which constantly improves the communication performance until convergence to an efficient suboptimal solution. The numerical results verifies the effectiveness of our proposed scheme for guaranteeing sensing QoS, especially a satisfactory CRB for on-demand localization. The high profits of our obtained solution in improving the communication throughput were also illustrated. 

In this work, our proposed design framework for 
UAV-enabled ISAC can potentially facilitate more on-demand sensing investigations in ISAC networks. 
For localization services,
we have developed an optimization scheme 
accurately tackling the localization CRB, which can be implemented in various ISAC scenarios toward strict localization accuracy guarantee and promote the corresponding network designs.\vspace{-0.1cm}

\bibliographystyle{IEEEtran}


\end{document}